\documentclass[reprint,twocolumn,
superscriptaddress,
 amsmath,amssymb,
 aps,
prl,
]{revtex4-2}

\usepackage{graphicx}
\usepackage{subfigure}
\usepackage{dcolumn}
\usepackage{bm}
\usepackage{amsmath}
\usepackage{amsthm}
\usepackage{xcolor}
\usepackage{ulem}
\usepackage{pifont}
\usepackage{textcomp}
\usepackage[colorlinks, linkcolor = blue, anchorcolor = blue, citecolor = cyan]{hyperref}
\usepackage{amsfonts}
\usepackage{dsfont}
\usepackage{mathrsfs}
\usepackage{accents}

\makeatother

\newtheorem{thm}{Theorem}
\newtheorem{lem}[thm]{Lemma}
\newtheorem{cor}[thm]{Corollary}
\newtheorem{ex}{Example}
\newtheorem{prop}[thm]{Proposition}

\newcommand{\N}{\mathcal{N}}
\newcommand{\M}{\mathcal{M}}
\newcommand{\B}{\mathcal{B}}

\newcommand{\tr}{\mathrm{Tr}}
\newcommand{\Q}{Q^{(1)}}

\newcommand{\bds}[1]{\vec{\boldsymbol #1}}

\begin{document}

\title{Extreme Capacities in Generalized Direct Sum Channels}

\author{Zhen Wu}
\affiliation{School of Mathematics and Statistics, Hainan University, Haikou, Hainan Province, 570228, China}
\affiliation{School of Mathematical Sciences, MOE-LSC, Shanghai Jiao Tong University, Shanghai, 200240, China}
\author{Si-Qi Zhou}
\email{zhousq9@mail.sysu.edu.cn}
\affiliation{School of Computer Science and Engineering, Sun Yat-sen University, Guangzhou 510006, China}
\affiliation{School of Mathematical Sciences, MOE-LSC, Shanghai Jiao Tong University, Shanghai, 200240, China}

\begin{abstract}
Quantum channel capacities play a central role in quantum Shannon theory, a formalism built upon rigorous coding theorems for noisy channels. Evaluating exact capacity values for general quantum channels remains intractable due to superadditivity. As a step toward understanding this phenomenon, we construct the generalized direct sum (GDS) channel, extending conventional direct sum channels through a direct sum structure in their Kraus operators. This construction forms the basis of the GDS framework, encompassing classes of channels with single-letter formula for quantum capacities and others exhibiting striking capacity features. The quantum capacity can approach zero yet display unbounded superadditivity combined with erasure channels. Private and classical capacities coincide and can become arbitrarily large, resulting in an unbounded gap with the quantum capacity. Providing a simpler and more intuitive approach, the framework deepens our understanding of quantum channel capacities.
\end{abstract}

\maketitle


\emph{Introduction.---} 
Channel capacity represents the ultimate physical limit on reliable communication. It quantifies the maximum rate at which information can be transmitted through a channel with asymptotic error-free performance. In the classical setting, this problem was resolved by Shannon in his seminal coding theorem~\cite{Shannon}, which laid the foundation of modern information theory. With the advent of quantum mechanics and the possibility of harnessing quantum resources for communication~\cite{NielsenChuang}, the notion of capacity was extended to the quantum domain, giving rise to quantum Shannon theory~\cite{Wilde,WatrousTQI}. At its core lie the quantum channel capacity theorems, which determine the ultimate rates at which different types of information can be faithfully transmitted over a quantum channel.

In the quantum regime, a channel $\B$ admits several distinct notions of capacity, reflecting different communication tasks. The most fundamental are the unassisted capacities: the quantum capacity $Q(\B)$~\cite{Hashing,InforTrans,LSD1,LSD2,LSD3} for quantum communication, the private capacity $P(\B)$~\cite{LSD3,RegExpP} for private classical communication, and the classical capacity $C(\B)$~\cite{classcapaHolevo,HSW1,HSW2} for classical information transmissions through quantum channels. These capacities not only characterize communication rates, but also capture the fundamental limits of other information theoretic tasks. In particular, the quantum capacity coincides with the highest rate of entanglement generation across the channel, and the private capacity with the maximum rate of secret-key generation~\cite{LSD2,BB84,E91}. More broadly, channel capacities quantify the amount of information that can be preserved in the presence of noise, thus providing a direct link to quantum error correction~\cite{QEC1,QEC2,CSS1,CSS2}. Together, these capacities specify the ultimate rates at which a quantum channel can transmit quantum, private, and classical information with asymptotically vanishing errors.

The mathematical characterization of quantum channel capacities is far more involved than in the classical case. Each of these capacities is expressed through a regularized formula (often called a multi-letter expression)~\cite{LSD1,LSD2,LSD3,RegExpP,HSW1,HSW2}, 
\begin{equation}\label{regexp}
    f(\B) = \lim_{m\rightarrow \infty} \frac{f^{(1)}(\B^{\otimes m})}{m}\,, 
\end{equation}
where $f$ denotes one of the capacities and $f^{(1)}$ is the corresponding one-shot capacity, namely the coherent information $\Q$, the private information $P^{(1)}$ and the Holevo capacity $C^{(1)}$ information.

A central obstacle in quantum Shannon theory is the phenomenon of \emph{superadditivity}: for two channels $\B_1$ and $\B_2$, one may have $f^{(1)}(\B_1\otimes \B_2) > f^{(1)}(\B_1) + f^{(1)}(\B_2)$, a phenomenon absent in classical theory~\cite{IEEEDeep,UniformAdd}. This effect has been observed for coherent information~\cite{SuperAddIcPauli,SuperAddIcPauli2,DephrasureChannel,ErrorThresholdsPauli,geneticalgorithms,IEEEPos,Filippov_2021,Singular}, private information~\cite{StructuredCodes,SuperAddIp1}, and Holevo information~\cite{Hastings2009,Shor2004}, as well as for the full quantum and private capacities~\cite{superactivation,Superact2,platypusPRL,wu2025,SuperAddP2,SuperAddP}. Whether the classical capacity is additive remains open. As a consequence, evaluating quantum channel capacities requires the regularized expression~\eqref{regexp}, a non-convex optimization over infinitely many variables for which no general method is known~\cite{UnboundedQ,UnboundP}.

Superadditivity is a defining feature of quantum channel capacities, yet there exist important cases where \emph{additivity} holds, $f^{(1)}(\B_1\otimes \B_2) = f^{(1)}(\B_1) + f^{(1)}(\B_2)$. In these instances, the regularized expression is no longer required, and the capacity is given exactly by the corresponding one-shot quantity, referred to as the single-letter formula. For classical communication, additivity of the Holevo information and the classical capacity has been established for several prominent channels when combined with arbitrary others~\cite{Additivity,Shor2002,Erasurechannel,AddCHadamard1,AddCHadamard2,King2003}. For quantum communication, the coherent information of degradable channels~\cite{LSD2,Structure} is additive when combined with other degradable channels. More generally, regularized less noisy channels, whose complementary channels have zero private capacity, admit single-letter formulas for both quantum and private capacities~\cite{Watanabe}.

\begin{figure}[!htb]
\begin{center}
\includegraphics[width=0.4\textwidth]{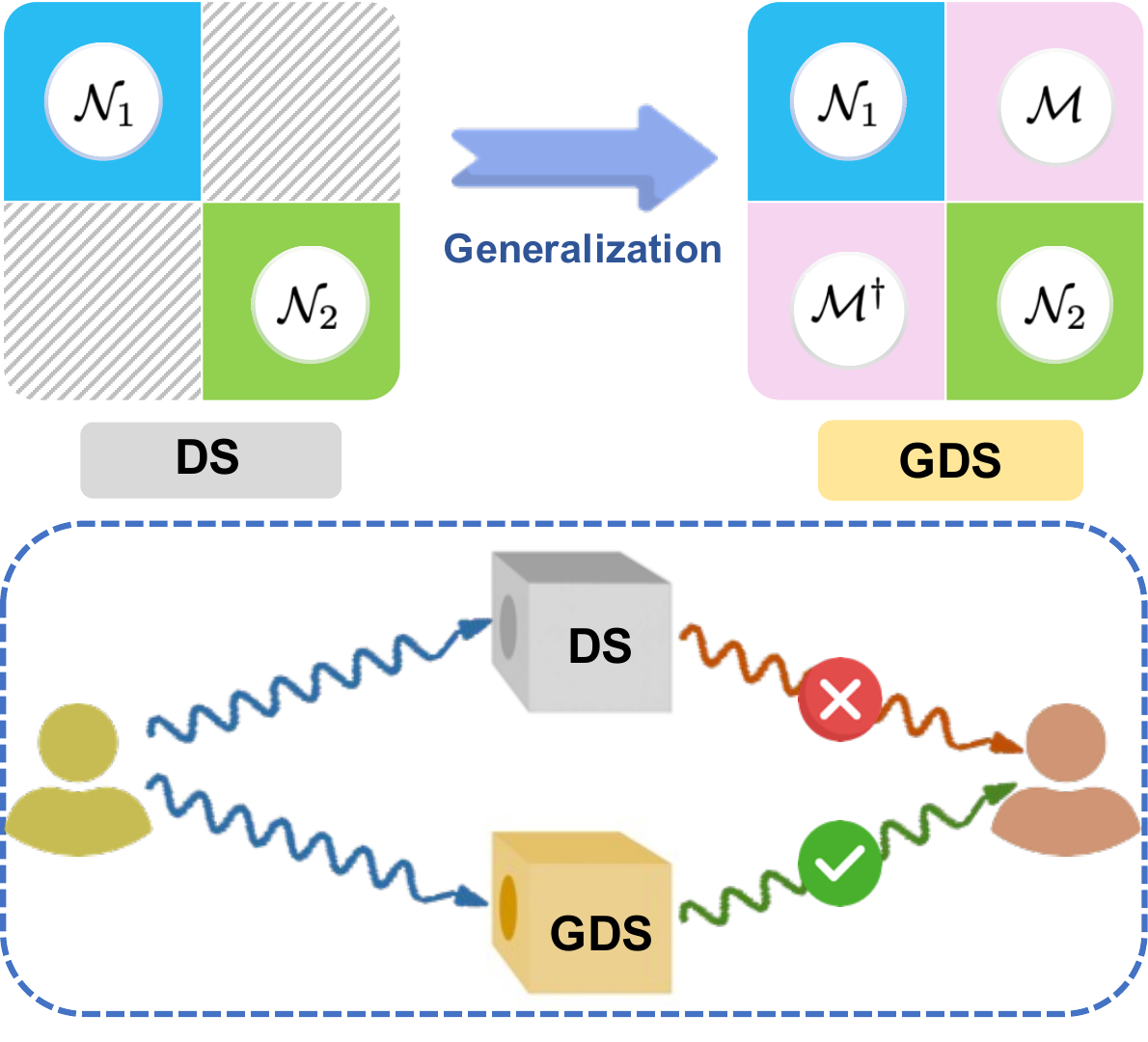}
\end{center}
\vspace{-7pt}
\caption{Structural contrast and communication advantage of generalized direct sum (GDS) over direct sum (DS) channels. The upper panel illustrates the structural distinction:  for simplicity, only two subchannels, $\N_1$ and $\N_2$, are shown. The DS channel exhibits a block-diagonal structure, whereas the GDS channel possess block coherent components $\M$ and $\M^\dagger$ linking the two subchannels. The lower panel highlights the resulting communication advantage: this structural distinction enables GDS channels to perform quantum communication tasks beyond the reach of DS channels.}\label{Figure1}
\vspace{-7pt}
\end{figure}

A fundamental ordering of channel capacities follows directly from their operational definitions: for any quantum channel $\B$, $Q(\B) \leq P(\B) \leq C(\B)$. Here we focus on the first inequality. Watanabe~\cite{Watanabe} provided a general sufficient condition for equality, showing that $Q(\B) = P(\B)$ whenever the complementary channel $\B^c$ has vanishing quantum capacity~\cite{IEEEPartialOrder}. Beyond this case, only a few explicit channels are known to exhibit a strict separation between $Q$ and $P$, such as Horodecki channels~\cite{Horodeckichannel2,Horodeckichannel3,HorodeckiChannel1,Horodeckichannel4}, which have vanishing quantum capacity and typically small private capacity; the half-rocket channel~\cite{Leung2014}, whose quantum capacity lies between $0.6$ and $1$ while its private capacity diverges; the platypus channel~\cite{wu2025,Lovasz,SDPBoundC,platypusTIT}, which achieves arbitrarily small but positive quantum capacity together with unit private capacity.

In this paper, we introduce the generalized direct sum (GDS) channel, defined through a sequence of subchannels $\{\N_i\}$. As shown in Fig.~\ref{Figure1}, the GDS channel, constructed from these subchannels, unlike ordinary direct sum (DS) channels, preserves the block structure while allowing coherent coupling between distinct subsystems, thereby enabling quantum communication processes beyond the reach of DS channels. We establish general bounds on the capacity of GDS channels and identify a class whose quantum capacity satisfies a single-letter formula, extending beyond all previously known additive coherent information. Furthermore, we construct a special instance generated from completely depolarizing subchannels. This channel possesses additive private and classical capacities, while its quantum capacity can be made arbitrarily small. It thus exhibits an unbounded separation between quantum and private capacities and demonstrates enormous superadditivity of quantum capacity. Compared with previous approaches, our GDS framework is markedly simpler and offers new insights into the content of quantum Shannon theory.

\emph{Preliminaries.---} 
The Hilbert space associated with system $A$ is denote by $\mathcal{H}_A$. A quantum channel $\mathcal{B}$ is a completely positive and trace-preserving (CPTP) linear map between operator algebras, which admits the Kraus form $\B(\rho) = \sum_i\, E_i \rho E_i^\dagger$ with $\sum_i E_i^\dagger E_i=\mathds{1}$. Equivalently, any channel can be represented via a unitary extension $U:A\to BE$ as $\B(\rho)  = \tr_E(U\rho U^\dagger)$, and the complementary channel is $\B^c(\rho) = \tr_B(U\rho U^\dagger)$. A channel $\mathcal{B}$ is degradable if there exists another channel $\mathcal{D}$ such that $\mathcal{D}\circ \B = \B^c$, in which case $\mathcal{B}^c$ is called antidegradable. Degradable channels admit a single-letter formula for the quantum capacity $\Q(\B) = Q(\B)$, whereas antidegradable channels have vanishing quantum capacity.

The channel capacity of interest is given by the regularized one-shot expression as Eq.~\eqref{regexp}. Specifically, $\Q(\B) = \max_\rho I_c(\rho,\B)$ is the coherent information of $\B$, where $I_c(\rho,\B) = S(\B(\rho)) - S(\B^c(\rho))$ and $S(\rho) = -\tr (\rho\log\rho)$ denotes the von Neumann entropy. For comparison, the private information and Holevo information of $\B$ are defined as 
$
    P^{(1)} = \max_{\mathfrak{E}} I_p(\mathfrak{E}, \B)\,, \,  C^{(1)} = \max_{\mathfrak{E}} \chi(\mathfrak{E}, \B)
$
where the maximization runs over all ensembles $\mathfrak{E} = \{p_i,\rho_i\}$ with average state $\bar{\rho} = \sum_i p_i\rho_i$, and
$I_p(\mathfrak{E},\B) = I_c(\bar{\rho}\,, \B) - \sum_i p_i I_c(\rho_i,\B)$, $\chi(\mathfrak{E}, \B) = S(\B(\bar{\rho}) - \sum_i p_i S(\B(\rho_i))$.

Recently, Chessa and Giovannetti~\cite{PCDSC} introduced the \textit{Partially Coherent Direct Sum} (PCDS) channel $\N$, generated by a set of subchannels $\{\N_i: A_i \to A_i\}_{i=0}^{n}$, where $\{A_i\}$ are mutually exclusive subsystems of the input system $A$. Formally, if $\{E_{k}^{(i)}\}_k$ is a Kraus representation of the subchannel $\N_i$, then the Kraus operator $\{E_k\}_k$ for the PCDS channel is $E_k = E_{k}^{(0)} \oplus \cdots \oplus E_k^{(n)}$. This direct sum structure of Kraus operators allows one to compute the quantum capacity of certain non-degradable PCDS channels~\cite{PCDSC}. Motivated by this observation, we relax the previous restriction that each subchannel $\N_i$ acts on the same input and output system, thereby introducing generalized direct sum channels.

\emph{Main results.---} 
We establish a general framework for generalized direct sum (GDS) channels, which exhibit several distinctive features with respect to quantum channel capacities. We show that the Kraus operators of a GDS channel necessarily take a direct sum form, so that the channel can be regarded as being generated by a collection of subchannels acting on distinct subsystems. A central consequence of this structural characterization is that degradability of a GDS channel is equivalent to the degradability of all its constituent subchannels. A further key observation lies in the optimal quantum state for the coherent information inherits a direct sum structure, enabling an exact determination of the quantum capacity in the degradable case. More generally, by analyzing the coherent information of generic GDS channels, we identify a class of quantum channels whose quantum capacity satisfies a single-letter formula and that lie beyond all previously known channels with additive coherent information. Taken together, these results delineate the GDS framework for quantum channel capacities.

To exemplify the framework, we construct a concrete class of GDS channels generated by $n+1$ distinct completely depolarizing subchannels. From the upper bound derived below, their quantum capacity can be made arbitrarily small, approaching zero. When combined with erasure channels, however, they exhibit a remarkable superadditivity. The joint quantum capacity reaches at least $\frac{1}{2}\log(n+1)$, which diverges with the number of subchannels, despite each individual capacity being vanishingly small. Moreover, the private and classical capacities coincide. Each admits a single-letter expression with exact value $\log(n+1)$, exposing an unbounded gap between quantum and private capacities. In contrast to previous approaches~\cite{SuperAddP,SuperAddP2,SuperAddIp1,Leung2014}, our framework dispenses with random unitaries and elaborate analysis, yielding a conceptually simple and fully explicit example.

In the following paragraphs we discuss our main results; see the Supplemental Material~\cite{SuppMaterial} for details.

\emph{Generalized direct sum channels.---}
Consider a quantum system $A$ with Hilbert space $\mathcal{H}_A$ admitting the direct sum decomposition $\mathcal{H}_A = \bigoplus_{i=0}^n \mathcal{H}_{A_i}$. Any linear operator $O$ on $\mathcal{H}_A$ can be written in block form $O = (O_{ij})_{i,j=0}^n$, where $O_{ij}: \mathcal{H}_{A_i} \to \mathcal{H}_{A_j}$. We call a channel $\N: A\to B$ to be a \textit{generalized direct sum} channel if it preserves the block structure between the input and output system. Explicitly, for any operator $O$, the output operator is
\begin{equation}\label{OutputOpe}
    \N(O) = \begin{pmatrix}
        \N_0(O_{00}) &  \cdots & \M_{0n}(O_{0n}) \\
        \vdots & \ddots & \vdots \\
        \M_{n0}(O_{n0}) & \cdots & \N_n(O_{nn})
    \end{pmatrix}
\end{equation}
where $\{\N_i\}$ act on the diagonal blocks, while $\M_{ij}$ act on the off-diagonal blocks. Analogous to the PCDS channels~\cite{PCDSC}, the CPTP condition for a GDS channel $\N$ is equivalent to requiring that its Kraus operators respect the same direct sum structure. The following lemma provides a precise characterization.

\begin{lem}
    A quantum channel $\N$ with Kraus operators $\{E_k\}$ is a GDS channel if and only if 
        \begin{equation}\label{KrausOpe}
    E_k = \bigoplus_{i=0}^n \; E_{k}^{(i)}\,.
\end{equation}
where $\{E_{k}^{(i)}\}$ is the Kraus operators of the subchannel $\N_i$, and $\M_{ij}(\cdot) = \sum_k E_k^{(i)} \cdot\, (E_{k}^{(j)})^\dagger$ with $\M_{ji} = \M_{ij}^\dagger$.
\end{lem}

The Kraus operators of $\N$ are fully determined by those of the subchannels $\{\N_i\}$; accordingly, we say that a GDS channel is generated by the Kraus operator $\{E_k^{(i)}\}$ of $\{\N_i\}$. As a consequence of the direct sum structure in Eq.~\eqref{KrausOpe}, its complementary channel is $\N^c(O) = \sum_{i=0}^n \N_i^c (O_i)\,.$ It holds that a GDS channel $\N$ is degradable if and only if all subchannels $\N_i$ are also degradable. In contrast, even when every $\N_i$ is antidegradable, the channel itself need not be. Furthermore, Eq.~\eqref{KrausOpe} yields the direct sum structure for optimal states of the coherent information of generic GDS channels, providing the upper and lower bounds for coherent information of a GDS channel. By utilizing antidegradable subchannels $\{\N_i\}_{i=0}^n$, these two bounds can be equal, thereby providing a class of GDS channels that satisfy single-letter formula for quantum capacity, which fall outside of any known types of channels with additive coherent information~\cite{Watanabe,IEEEPartialOrder,UniformAdd}, see more discussion in~\cite{SuppMaterial}.

\begin{figure*}[t]
\includegraphics[width=0.9\textwidth]{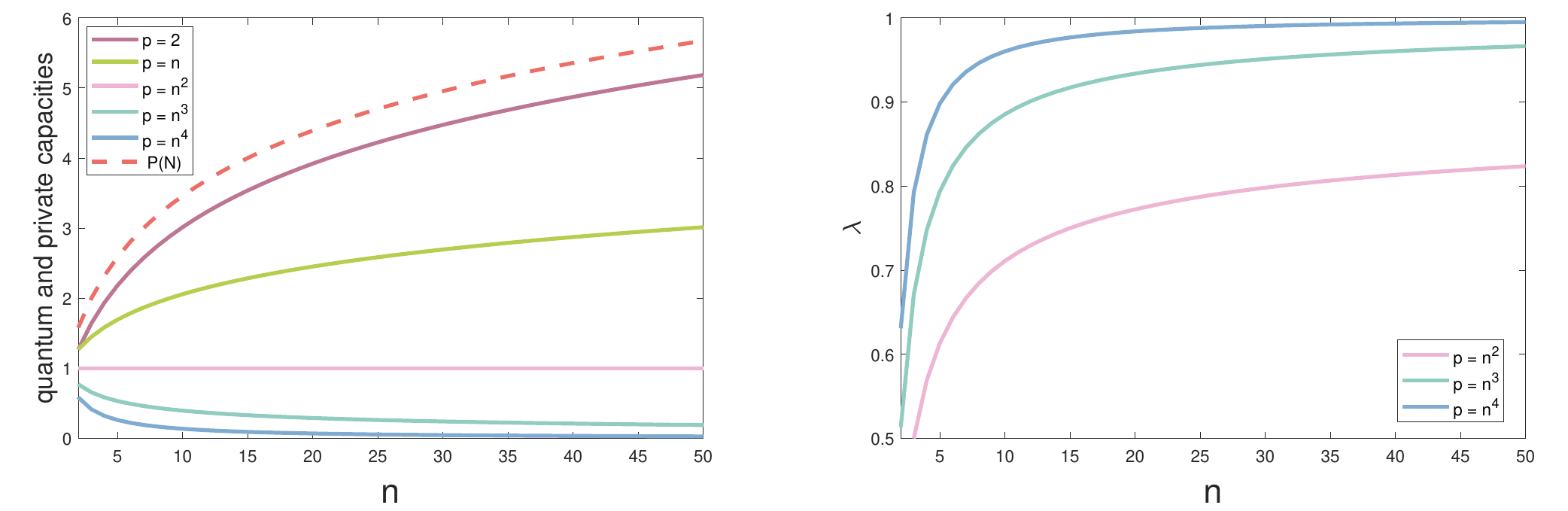}
\caption{Channel capacities analysis for the GDS channel. (left) Solid lines denote upper bounds on the quantum capacity of the GDS channel $\N$ for various $p$. The dash line indicates the exact private capacity $\log(n+1)$. The gap between quantum and private capacity diverges as $n$ or $p$ increases. (right) Maximal value of $\lambda$ for which the quantum capacity of $\N$ combined with $\mathcal{E}_\lambda$ exhibits superadditivity.}
\label{fig1}
\vspace{-7pt}
\end{figure*}

\emph{Upper bound of quantum capacity.---} 
We establish an analytic upper bound on the quantum capacity of a generic GDS channel $\N$. The central tool is the ``transposition bound''~\cite{TranspositionBound}, which for any channel $\B: A' \to B$ guarantees $Q(\B) \leq \log ||\mathcal{T}\circ \B||_{\diamond}$, with $\mathcal{T}$ the transposition map and $||\cdot||_{\diamond}$ the diamond norm. We next derive the upper bound by relating the transposition bounds of subchannels $\N_i$. For a matrix $M$, its absolute value is defined by $|M| = \sqrt{M^\dagger M}$. With this notation, the following result holds.

\begin{thm}\label{UpBdQN}
  Let $\N$ be a GDS channel generated by subchannels $\{\N_i\}_{i=0}^n$ with Kraus operators $\{E_k^{(i)}\}$ and $\mathscr{C}_{\M_{ij}}$ is the Choi matrix of $\M_{ij}$, then the quantum capacity of $\N$ satisfies the upper bound
  \begin{equation}\label{upperBdQ}
    \begin{aligned}
      Q(\N) & \leq \log \max_i \Big[||\mathcal{T}\circ \N_i||_{\diamond}  \\
      & \qquad\quad + \sum_{j\neq i} \sqrt{||\tr_B(|\mathscr{C}_{\M_{ij}}^{T_B}|)||_\infty\,||\tr_B(|\mathscr{C}_{\M_{ji}}^{T_B}|)||_\infty}\Big] \,,
    \end{aligned}
  \end{equation}
 where $||\cdot||_\infty$ denotes the spectral norm.
\end{thm}

Theorem~\ref{UpBdQN} demonstrates that the quantum capacity is influenced not only by the subchannels $\N_i$ but also by the off-diagonal maps $\M_{ij}$ associated with the off-diagonal blocks in Eq.~\eqref{OutputOpe}. It is natural to expect this behavior: unlike the ordinary direct sum of channels~\cite{DirectSumC}, the GDS construction incorporates block coherent component $\M_{ij}$, which reveals that the GDS construction can leverage block coherence of the input quantum state to achieve higher communication rates, as discussed in~\cite{SuppMaterial}. Furthermore, although direct sum of channels $\Phi_1\oplus \Phi_2$ have played a significant role in proving the superadditivity of Holevo information~\cite{Hastings2009}, such superadditivity actually comes from the tensor product for these two channels $\Phi_1\otimes \Phi_2$ rather than the direct sum structure itself. Since a direct sum structure corresponds to block coherence~\cite{aberg06,CohePRL1}, which is recognized as a quantum resource~\cite{CoheRMP,CohePRL2}, it's reasonable to believe that GDS construction can offer new perspectives into the additivity and superadditivity of channel capacities. As an illustration, we now present a notable example for which the channel capacities can be computed explicitly.

\emph{Completely depolarizing subchannels.---}
We introduce a family of GDS channels $\N$, characterized by integers $p>1$ and $n$, and constructed from a set of completely depolarizing subchannels $\{\N_i\}_{i=0}^{n}$. Each subchannel $\N_i$ maps an input state $\rho_i$ on subsystem $A_i$ to the completely mixed state on the output subsystem $B_i$, where the dimensions satisfy $d_{A_i} = p^i$ and $d_{B_i} = p^{n - i}$. The Kraus operators associated with $\N_i$ take the form $\big\{\frac{1}{\sqrt{d_{B_i}}}|s\rangle\langle t|\,, s = 0,\ldots, p^i-1, t=0,\ldots, p^{n-i}-1\big\}$, which give rise to the channel representation
\begin{equation}\label{CDC}
\N_i: A_i \to B_i,\,\quad \N_i(\rho_i) = \tr(\rho_i)\, \frac{\mathds{1}_{B_i}}{d_{B_i}}\,.
\end{equation}
Each subchannel $\N_i$ is completely noisy and thus has vanishing capacity. Remarkably, when combined into a GDS channel $\mathcal{N}$, the resulting channel exhibits strictly positive capacities. A straightforward calculation shows that
$||\tr_B(|\mathscr{C}_{\M_{ij}}^{T_B}|)||_\infty\,||\tr_B(|\mathscr{C}_{\M_{ji}}^{T_B}|)||_\infty = p^{-|j-i|}\,.$
Consequently, Theorem~\ref{UpBdQN} implies that the quantum capacity of the GDS channel generated by the depolarizing subchannels in Eq.~\eqref{CDC} is bounded above by
\begin{equation}\label{UpQCDC}
    Q(\N) \leq \log \Big(1+\frac{n}{\sqrt{p}}\Big)\,.
\end{equation}

The upper bounds on the quantum capacity of $\N$ for varying $p = p(n)$ are shown in Fig.~\ref{fig1}. The figure also highlights the exact value $\log(n+1)$ for the private capacity of $\N$, displaying the large separation between the quantum and private capacities, together with the maximum erasure probability for which superadditivity of quantum capacity persists when $\N$ is combined with erasure channels, as discussed below. A fundamental lower bound on the coherent information of GDS channels was established in Ref.~\cite{SuppMaterial}, namely $\Q(\N) \geq \log(n+1)/p^n > 0$. Since the integers $p$ and $n$ are arbitrary, for any positive number $\beta$ one can choose $p\gg n$ such that $0<Q(\N) <\beta$, as illustrated in Fig.~\ref{fig1}. In particular, the ``general platypus channel''~\cite{platypusTIT,wu2025} arises as a special case of a GDS channel constructed from two discard-prepare subchannels: $\N_0: A_0\to B_0,\, \N_0(\cdot) = \sigma$ with $d_{A_0} = d_{B_1} = d$, and $\N_1: A_1 \to B_1,\, \N_1(\cdot) = |0\rangle\langle 0|_{B_1}$ with $d_{B_0} = d_{A_1} = 1$. Therefore, GDS channels extend the class of platypus channels while preserving the feature of possessing arbitrarily small yet positive quantum capacity.

Furthermore, the private and classical capacities of this specific GDS channel admit exact valuation by adapting the methods of Refs.~\cite{platypusTIT,wu2025}. Consider the ensemble $\mathfrak{E} = \{\frac{1}{n+1}, \frac{\mathds{1}_{A_i}}{d_{A_i}} \oplus 0_{A_i^c}\}_{i=0}^n$ where $\frac{\mathds{1}_{A_i}}{d_{A_i}}$ is supported on subsystem $A_i$ and $0_{A_i^c}$ denotes the zero operator on all remaining subsystems. For this ensemble, the private information of $\N$ satisfies 
$ P^{(1)}(\N) \geq I_p(\mathfrak{E},\N) = \log (n+1). $
An upper bound on the classical capacity is provided in Refs.~\cite{SDPBoundC,SDPBoundC2}, which guarantees that for any channel $\B$ its classical capacity satisfies $ C(\B) \leq \log \beta(\B)$. Through direct construction, it can be derived that $\beta(\N) \leq n+1$, leading to the following capacity bounds for $\N$.

\begin{thm}
    Let $\N$ be a GDS channel generated by completely depolarizing subchannels $\{\N_i\}_{i=0}^n$ defined in Eq.~\eqref{CDC}. Then the quantum channel capacity of $\N$ satisfy
    \begin{widetext}
        \begin{equation}\label{CapacityCDC}
         \frac{\log(n+1)}{p^n} \leq \Q(\N) \leq Q(\N) \leq \log \Big(1+\frac{n}{\sqrt{p}}\Big) < \log (n+1) = P^{(1)}(\N) = C^{(1)}(\N) = P(\N) = C(\N)\,.
    \end{equation}
    \end{widetext}
\end{thm}

By choosing a suitable scaling of $p = p(n)$, for instance $p(n) = n^4$, the gap between quantum and private diverges as $n$ grows. A similar extensive separation was previously observed in the ``half-rocket channels''~\cite{Leung2014}, which maintain a quantum capacity exceeding $0.6$; in contrast, the quantum capacity of the present GDS channel can be made arbitrarily close to 0. Moreover, whereas the half-rocket construction relies on random unitaries and does not admit a simple closed form, the GDS channels introduced here provide an explicit example with a transparent expression. In addition, restricting the complementary channel $\N^c$ to subsystem $A_n$ yields an identical channel, implying that $Q(\N^c) \geq n\log p$. Consequently, the GDS channel does not fall into any previously known class of channels with additive private and Holevo information.

To demonstrate the superadditivity of the quantum capacity for the GDS channel $\N$ generated by subchannels $\{\N_i\}$ in Eq.~\eqref{CDC}, we combine $\N$ with the erasure channel $\mathcal{E}_\lambda (\rho) = (1-\lambda) \rho + \lambda \tr(\rho)|e\rangle\langle e|$, whose quantum capacity is $Q(\mathcal{E}_\lambda) = \max\{(1-2\lambda)\log d_{A}, 0\}$~\cite{Erasurechannel}.
Consider the quantum states $\{\mathds{1}_{A_i}/d_{A_i}\}_{i=0}^n$ corresponds to their purifications $\{|\Phi_i\rangle\langle\Phi_i|_{A_iA_i'}\}$, and define the bipartite state
\begin{equation}
    \omega_{AA'} = \frac{1}{n+1} \bigoplus_{i=0}^n\, |\Phi_i\rangle\langle\Phi_i|_{A_iA_i'}\,.
\end{equation}
Using the output of the complementary channel $\N^c$, the coherent information $I_c(\omega_{AA'},\N\otimes \mathcal{E}_\lambda)$ with respect to this input state as
\begin{equation}
    I_c(\omega_{AA'},\N\otimes \mathcal{E}_\lambda) = (1-\lambda)\log (n+1) \,.
\end{equation}
Combining this with the upper bound on $Q(\N)$ from Eq.~\eqref{UpQCDC} shows that, for $\lambda \geq 1/2$, there exist values of $p$ for which the quantum capacity of the GDS $\N$ combined with $\mathcal{E}_\lambda$ is superadditive,
$ Q(\N\otimes \mathcal{E}_{\lambda}) > Q(\N) + Q(\mathcal{E}_\lambda)\,,$
as illustrated in Fig.~\ref{fig1}. Notably, the range of $\lambda$ increases with $p$; for fixed $p$, the maximum $\lambda$ guaranteeing the superadditivity also gradually approaches $1$ when $n$ increases.

\emph{Discussion.---} 
The generalized direct sum (GDS) channel introduced in this paper employs a direct sum structure in its Kraus operators, thereby exhibiting estimable quantum channel capacities and forming a class of channels with a single-letter expression for the quantum capacity. GDS channels generated by completely depolarizing subchannels demonstrate both an arbitrarily large separation between quantum and private capacities and a pronounced superadditivity of the quantum capacity. Combining several completely depolarizing channels within the GDS framework enables quantum, private, and classical communication, outshining the quantum switch~\cite{Switch2,CSwitchPRL,QSWofCDC,WuSwitch} and coherent control~\cite{CoheControl} for these subchannels. For a fixed number of completely depolarizing subchannels, the classical capacity of the GDS channel remains independent of the input dimension, unlike the channels constructed via the quantum switch and coherent control, whose classical capacities vanish as the input dimension increases.

The GDS framework, in contrast to the quantum switch that relies on indefinite causal order, leverages a direct sum structure associated with coherence. Different quantum resources are expected to provide distinct communication advantages, enriching the content of quantum Shannon theory. We anticipate that the GDS framework will not only stimulate further developments in quantum Shannon theory but also find experimental realization and practical applications in broader areas of quantum information.


\emph{Acknowledgements.---}
Z.W. thanks Peixue Wu for helpful discussions.


\bibliography{reference}


\clearpage\clearpage
\newpage

\title{Methods}
\author{testing}

\maketitle
\onecolumngrid

\begin{center}\large \textbf{Extreme Capacities in Generalized Direct Sum Channels} \\
\textbf{--- Supplementary Material ---}\\
\end{center}



\onecolumngrid

The Supplemental Material is organized as follows. In Section I, we establish several basic properties of generalized direct sum (GDS) channels, including the structure of their Kraus operators and degradability. In Section II, we derive bounds on the coherent, private, and Holevo information, as well as the quantum capacity of a GDS channel. Subsequently, in section III, building on the bounds for the coherent information of a GDS channel provided in section II, we propose a new class of quantum channels with a single-letter formula for quantum capacity. Finally, section IV demonstrates the attractive capacity features of the special GDS channel in the main text.

\tableofcontents
\vspace{1cm}

\newpage

\section{I. Basic properties of generalized direct sum channels}

Let $\mathcal{H}_A$ be the Hilbert space corresponding to the quantum system $A$. A quantum channel from the input system $A$ to the output system $B$ is a completely positive and trace preserving linear map $\mathcal{N}: A \to B$ admitting the Kraus representation
\begin{equation}
  \N(\rho) = \sum_i\, E_i \rho E_i^\dagger\,,
\end{equation}
where the operators $\{E_i\}$ satisfy $\sum_i E_i^\dagger E_i = \mathds{1}$. The dual channel $\N^\dagger$ of $\N$ is a unital quantum channel with Kraus operators $\{E_i^\dagger\}$. Equivalently, the quantum channel has a Stinespring representation: there exists a unitary operator $V:A\to BE$ such that $\N(\rho) = \tr_E(V\rho V^\dagger)$, where $E$ denotes the environment system. In particular, the unitary can be chosen as $V = \sum_i E_i\otimes |i\rangle_E$.

\subsection{A. Kraus representation of generalized direct sum (GDS) channels}

The GDS channel generalizes the partially coherent direct sum channel introduced in~\cite{PCDSC}, and it no longer restricts the input and output operators of subchannels to be in the same space. Such a generalization is quite natural and can have a wider range of applications. In the main text, we give a special channel generated by completely depolarizing subchannels, which is a GDS channel but not a partially coherent direct sum channel.

Now let the input system $\mathcal{H}_A$ and the output system $\mathcal{H}_B$ admit the direct sum structure
\begin{equation*}
  \mathcal{H}_A = \mathcal{H}_{A_0} \oplus \cdots \oplus \mathcal{H}_{A_n}\,, \quad \mathcal{H}_B = \mathcal{H}_{B_0} \oplus \cdots \oplus  \mathcal{H}_{B_n}
\end{equation*}
where $\mathcal{H}_{A_i}$ are non-trivial subspace of $\mathcal{H}_{A}$ with dimensions $d_{A_i}$, and $\mathcal{H}_{B_j}$ are non-trivial subspace of $\mathcal{H}_B$ with dimensions $d_{B_j}$. The GDS channel $\N$ is a quantum channel mapping the input system $A$ to the output system $B$ preserving the direct sum structure. Formally, the Hermitian operation $O$ can be written as $O = (O_{ij})$, where each $O_{ij}$ is a linear operator from $\mathcal{H}_{A_i}$ to $\mathcal{H}_{A_j}$. A channel $\N$ is called a \textit{generalized direct sum} channel if it preserves the block structure between the input and output systems; that is, for every operator $O$, the output operator can be written as
\begin{equation}
    \N(O) = \begin{pmatrix}
        \N_0(O_{00}) & \M_{01}(O_{01}) & \cdots & \M_{0n}(O_{0n}) \\
        \M_{10}(O_{10}) & \N_1(O_{11}) & \cdots & \M_{1n}(O_{1n}) \\
        \vdots & \vdots & \ddots & \vdots \\
        \M_{n0}(O_{n0}) & \M_{n1}(O_{n1}) & \cdots & \N_n(O_{nn})
    \end{pmatrix}.
\end{equation}
The following theorem describes the structure of Kraus operators of a GDS channel $\N$ and the associated linear map $\M_{ij}$.
\begin{thm}\label{KrausSM}
  A quantum channel $\N$ with Kraus operators $\{E_k\}$ is a GDS channel if and only if
  \begin{equation}\label{eqKraus}
    E_k = E_{k}^{(0)}\oplus \cdots \oplus  E_{k}^{(n)}\,,
  \end{equation}
 for all $k$, where $\{E_{k}^{(i)}\}$ are the Kraus operator of quantum channels $\N_i$, and $\M_{ij}(\cdot) = \sum_k E_{k}^{(i)} \cdot E_{k}^{(j),\dagger}$ is a linear map. 
\end{thm}
\begin{proof}
  Suppose that the Kraus operators $\{E_k\}$ of a quantum channel $\N$ satisfy Eq.~\eqref{eqKraus}. Then, for every index $i = 0,\cdots,n$, we have $\sum_k E^{(i),\dagger}_k E_k^{(i)} = \mathds{1}$. Given an input operator $O$, the output operator $\N(O)$ is 
  \begin{equation*}
    \begin{aligned}
      \N(O) & = \sum_k \begin{pmatrix}
        E_{k}^{(0)} & & \\
        & \ddots & \\
        &  & E_{k}^{(n)}
        \end{pmatrix}
        \begin{pmatrix}
        O_{00} & \cdots & O_{0n} \\
        \vdots & \ddots & \vdots \\
        O_{n0} & \cdots & O_{nn}
    \end{pmatrix}
     \begin{pmatrix}
         E_{k}^{(0),\dagger} & & \\
        & \ddots & \\
        &  & E_{k}^{(n),\dagger}
      \end{pmatrix} \\
      & = \sum_k \begin{pmatrix}
        E_{k}^{(0)} O_{00} E_{k}^{(0),\dagger} & \cdots & E_{k}^{(0)} O_{0n} E_{k}^{(n),\dagger} \\
        \vdots & \ddots &\vdots \\
         E_{k}^{(n)} O_{n0} E_{k}^{(0),\dagger} & \cdots & E_{k}^{(n)} O_{nn} E_{k}^{(n),\dagger} \\
      \end{pmatrix} \\
      & =  \begin{pmatrix}
        \N_0(O_{00}) & \M_{01}(O_{01}) & \cdots & \M_{0n}(O_{0n}) \\
        \M_{10}(O_{10}) & \N_1(O_{11}) & \cdots & \M_{1n}(O_{1n}) \\
        \vdots & \vdots & \ddots & \vdots \\
        \M_{n0}(O_{n0}) & \M_{n1}(O_{n1}) & \cdots & \N_n(O_{nn})
    \end{pmatrix},
    \end{aligned}
  \end{equation*}
  where $\N_i(O_{ii}) = \sum_k E_{k}^{(i)} O_{ii} E_{k}^{(i),\dagger}$ are quantum channels and $\M_{ij}(O_{ij}) = \sum_k E_{k}^{(i)} O_{ij} E_{k}^{(j),\dagger}$ are linear maps. Thus, the quantum channel $\N$ preserves the direct sum decomposition.

  Now let $\N$ be a GDS channel with Kraus operators $\{E_k\}$. Define $P_{A_i}$ as the orthogonal projector onto the subspace $\mathcal{H}_{A_i}$ of system $A$, and $P_{B_j}$ as the orthogonal projector onto the subspace $\mathcal{H}_{B_j}$ of system $B$. Let the input operator be $O = P_{A_i}$. If $i\neq j$, we can find 
  \begin{equation*}
    \begin{aligned}
      0 = P_{B_j} \N(P_{A_i}) P_{B_j}^\dagger = \sum_k P_{B_j}E_{k} P_{A_i} P_{A_i}^\dagger E_{k}^\dagger P_{B_j}^\dagger  = \sum_{k} \big|P_{B_j}E_kP_{A_i}\big|^2\,,
    \end{aligned}
  \end{equation*}
  which implies that $P_{B_j}E_kP_{A_i} = 0$ for all $E_k$. Totally, taking over all $i\neq j$, we obtain the direct sum structure of $E_k$ as given in Eq.~\eqref{eqKraus}.
\end{proof}

The remarkable example called ``platypus channels'' introduced in~\cite{platypusTIT} is in fact a special GDS channel. The generalized platypus channel $\mathcal{O}_{\bds{\mu}}: A\to B$ with $d_A = d_B = d+1$ is defined by the following unitary $H$
\begin{equation*}
    H |0\rangle = \sum_{j=0}^{d-1} \sqrt{\mu_j} |j\rangle \otimes |j\rangle \,, \quad H |i\rangle = |d\rangle\otimes |i-1\rangle\,,\quad \text{for}\, 1\leq i \leq d\,,
\end{equation*}
where $\bds{\mu} = (\mu_0,\cdots,\mu_{d-1})$ is a probability vector. The Kraus operators of $\mathcal{O}_{\bds{\mu}}$ are given by $E_k = \sqrt{\mu_k}|k\rangle\langle 0| + |d\rangle\langle k+1|$ for all $k = 0,\cdots,d-1$. These Kraus operators $E_k$ can be rewritten as $E_k = E^{(0)}_k \oplus E^{(1)}_k$ with $E^{(0)}_k = \sqrt{\mu_k}|k\rangle_{B_0}\langle 0|_{A_0}$ and $E_k^{(1)} = |0\rangle_{B_1}\langle k|_{A_1}$, corresponding to two subchannels 
$$\N_0: A_0\to B_0,\quad \N_0(\rho_0) = \begin{pmatrix}
    \mu_0 & & \\
    & \ddots & \\
    & & \mu_{d-1}
\end{pmatrix}\,, \quad \N_1: A_1\to B_1,\quad \N_1(\rho_1) = |0\rangle\langle 0|$$ 
with $d_{A_0} = d_{B_1} = 1$, and $d_{B_0} = d_{A_1} = d$.

In general, neither GDS channels nor partially coherent direct sum channels are direct sum channels, since the first two classes of channels maintain coherent parts. On the other hand, by adding zero operators to the Kraus operators of the subchannels, one can construct a direct sum channel, which implies that a direct sum channel is a special type of GDS channel.

\subsection{B. Complementary channels and degradability for GDS channels}

The complementary channel of a quantum channel $\N$ with Kraus operators $\{E_k\}$ is given by
\begin{equation*}
  \N^c(\rho) = \sum_{k,k'} \tr\big(E_k\rho E_{k'}^\dagger\big)\, |k\rangle\langle k'|\,.
\end{equation*}
Given a GDS channel $\N$, according to Theorem~\ref{KrausSM} about its Kraus operators, there exist quantum channels $\{\N_i\}_{i=0}^n$ such that each channel $\{\N_i\}$ has Kraus operators $\{E_k^{(i)}\}$ to generate the Kraus operators of the GDS channel $\N$. For every input operator $O$, the output operator of the complementary channel for the GDS channel $\N$ is 
\begin{equation}\label{complementary}
  \N^c(O) = \sum_{k,k'} \tr(E_k O E_{k'}^\dagger)\, |k\rangle\langle k'| = \sum_{k,k'} \tr\Big[\sum_{i=0}^n E_{k}^{(i)} O_{ii} E_{k'}^{(i),\dagger} \Big]\, |k\rangle\langle k'| = \sum_{i=0}^n \N_i^c(O_{ii}).
\end{equation}

Recall a quantum channel $\B$ is said to be \emph{degradable} if there exists another channel $\mathcal{W}$ such that $\mathcal{W}\circ \B = \B^c$, while $\B$ is called \emph{antidegradable} if $\B^c$ is degradable. If subchannels $\N_i$ are all degradable, that is there exist channels $\mathcal{W}_i$ such that $\mathcal{W}_i \circ \N_i = \N^c_i$, then one can define a new channel $\mathcal{W}(R) = \sum_{i=0}^n \mathcal{W}_i(R_{ii})$, since all $\mathcal{W}_i$ are quantum channels. Then, utilizing Eq.~\eqref{complementary}, which gives the output operator of the complementary channel for a GDS channel $\N$, we obtain the following result that describes the degradability condition for GDS channels.
\begin{thm}\label{degradablity}
  A GDS channel $\N$ is degradable if and only if the subchannels $\{\N_i\}$ are all degradable.
\end{thm} 
\begin{proof}
  Suppose that the subchannels $\{\N_i\}$ are all degradable. Then there exist channels $\mathcal{W}_i$ such that $\mathcal{W}_i\circ \N_i = \N_i^c$. For every operator $R = (R_{ij})$ on the system $B$, we can define the channel $\mathcal{W}$ as
  \[
  \mathcal{W}(R) = \mathcal{W}_0(R_{00}) + \cdots + \mathcal{W}_n(R_{nn})\,.
  \]
  Using Eq.~\eqref{complementary}, for every operator $O$ on system $A$, we have 
  \begin{align*}
    \mathcal{W} \circ \N(O) & = \mathcal{W} \begin{pmatrix}
        \N_0(O_{00}) & \M_{01}(O_{01}) & \cdots & \M_{0n}(O_{0n}) \\
        \M_{10}(O_{10}) & \N_1(O_{11}) & \cdots & \M_{1n}(O_{1n}) \\
        \vdots & \vdots & \ddots & \vdots \\
        \M_{n0}(O_{n0}) & \M_{n1}(O_{n1}) & \cdots & \N_n(O_{nn})
    \end{pmatrix} \\
  & = \mathcal{W}_0\circ\N_0(O_{00}) + \cdots + \mathcal{W}_n\circ \N_n(O_{nn}) \\
  & = \N_0^c(O_{00}) + \cdots + \N_n^c(O_{nn}) = \N^c (O)\,,
  \end{align*}
  which implies the degradability of $\N$\,.

  Now if $\N$ is degradable, there exists a channel $\mathcal{W}$ so that $\mathcal{W}\circ\N = \N^c$\,. According to the expression of output operators of the channel $\N$, the channel $\mathcal{W}_i :=  \mathcal{W}|_{\mathcal{H}_{B_i}} = \mathcal{W}(P_{B_i}\cdot P_{B_i})$ acts from system $\mathcal{B}_i$ to the environment system $E$. Since every operator $R_{ii}$ on system $B_i$ can be embedded into the full system $B$ as $R$ satisfying $R_{ii} = P_{B_i} R P_{B_i}$. Thus
  \[
  \mathcal{W}_i \circ \N_i(O_{ii}) =  \mathcal{W} \begin{pmatrix}
  0 & & & &\\
  & \ddots & & &\\
  & & \N_i(O_{ii}) & & \\
  & & & \ddots &\\
  & & & & 0
  \end{pmatrix} = \N_i^c(O_{ii})\,,
  \]
  which shows that $\N_i$ is degradable for all $i = 0,\cdots, n$.
\end{proof}

One might wonder whether the antidegradability of a GDS channel $\N$ is also determined by the antidegradability of its subchannels $\N_i$. We give a negative answer to this question: even when all $\{\N_i\}$ are completely depolarizing channels and thus cannot transmit any information, the GDS channel $\N$ does have a positive quantum capacity.

\subsection{C. Dependence on implementations of subchannels}

It is well-known that the Kraus operators of a given quantum channel are not unique. If two Kraus operator sets $\{K_i\}$ and $\{K'_i\}$ give the same channel, there exists a unitary operator $U$ such that $K'_i = \sum_j U_{ij} K_j$ for all $i$. Suppose that $\{E_k^{(i)}\}$ are the Kraus operators corresponding to the Stinespring unitary  $V^{(i)} = \sum_k E_k^{(i)}\otimes |k\rangle$ of $\N_i$. Then, the Stinespring unitary of the corresponding GDS channel $\N$ can be chosen as 
\[
    V = V^{(0)} \oplus \cdots \oplus V^{(n)}\,.
\]
Now suppose that $\{F_{s}^{(i)}\}$ are also the Kraus operators of $\N_i$ and $U^{(i)}$ is the unitary so that $F_{s}^{(i)} = \sum_k U^{(i)}_{s k} E_{k}^{(i)}$. These Kraus operators of $\N_i$ can generate the other GDS channel $\widetilde{\N}$ with Kraus operators as $F_{s} = F_{s}^{(0)}\oplus \cdots \oplus  F_{s}^{(n)}$  according to Theorem~\ref{KrausSM}, which gives rise to the Stinespring unitary as
\[
\widetilde{V} = (\mathds{1}_{B_0} \otimes U^{(0)}) V^{(0)} \oplus \cdots \oplus (\mathds{1}_{B_n} \otimes U^{(n)}) V^{(n)} = \Big[ \bigoplus_{i=0}^n (\mathds{1}_{B_i}\otimes U^{(i)}) \Big]\, V\,.
\]
Thus the channel $\widetilde{\N}(\rho) = \tr_E(\widetilde{V}\rho \widetilde{V}^\dagger)$ is \emph{not} equal to the GDS channel $\N$ in general. The complementary channel of the GDS channel $\widetilde{\N}$ is formulated as 
\begin{align*}
\widetilde{\N}^c(O)  = \sum_{s,s'} \tr(F_s O F_{s'}^\dagger)|s\rangle\langle s'| = \sum_{s,s'} \Big[\sum_{i=0}^n \tr(F_s^{(i)}O_{ii} F_{s'}^{(i),\dagger})\Big]\, |s\rangle\langle s'| = \sum_{i=0}^n U^{(i)} \N_i^c(O_{ii}) U^{(i),\dagger}\,.
\end{align*}
Particularly, only when all these $U^{(i)} = U$ are identical, the Stinespring unitary $\widetilde{V} = (\mathds{1}\otimes U)\ V$, and hence $\widetilde{\N} = \N$. Otherwise, these two GDS channels are not equivalent, even in the information-theoretic sense~\cite{ChanEqui} (but their channel capacities may be identical).

Therefore, for a given set of subchannels $\{\N_i\}$, different choices of Kraus operators for these subchannels may give rise to different GDS channels. In this work, we assume that the subchannels have different Kraus operators; that is, for every $k$, $E_k^{(i)} \neq E_k^{(j)}$ for all $i,j$, even if the Kraus operators differ only by their ordering. If all these subchannels have the same Kraus operators $\{E_k^{(i)}\}$, the GDS channel $\N = \mathcal{I}_{n+1} \otimes \N_1$ is trivial. If some of the subchannels are with same Kraus operators and others are different, then the GDS channel $\N$ generated by all these subchannels is $\N = \mathcal{I}\otimes \N'$, where $\N'$ is the GDS channel generated by the subchannels with different Kraus operators. Here, we give an example of the same subchannels with different Kraus operators to generate different GDS channels.
\begin{ex}
    Let $\N_1 = \N_2 = \Delta: A\to B\,, \Delta(\rho) = \sum_{i=0}^1 |i\rangle\langle i|\rho|i\rangle\langle i|$ be two identical qubit completely dephasing channels from the qubit input system $A$ to the qubit output system $B$. There are two possible Kraus operator sets of the subchannel:
    \[
     \left\{E_0 = \begin{pmatrix}
         \frac{\sqrt{2}}{2} &   \\
           & \frac{\sqrt{2}}{2}
    \end{pmatrix},E_1 = \begin{pmatrix}
         \frac{\sqrt{2}}{2} &  \\
          & -\frac{\sqrt{2}}{2}
    \end{pmatrix} \right\}\,,\qquad  \left\{F_0 = \begin{pmatrix}
        1 &   \\
         & 0
    \end{pmatrix},F_1 = \begin{pmatrix}
         0 &  \\
         & 1
    \end{pmatrix} \right\}\,.
    \]
    \begin{itemize}
        \item[(1).] $\{E_0,E_1\}$ and $\{E_0,E_1\}$ can generate a GDS channel $\B_1(\sigma) = \mathcal{I}\otimes \Delta(\sigma)$ for all ququart input states $\sigma$.
     \item[(2).] $\{E_0,E_1\}$ and $\{E_1,E_0\}$ can generate a GDS channel $\B_2$, in the expression as 
    \[
    \B_2(\sigma) =\begin{pmatrix}
        \sigma_{00} & & \sigma_{02} &  \\
        & \sigma_{11} &  & -\sigma_{13} \\
        \sigma_{20} &  & \sigma_{22} & \\
         & -\sigma_{31} & & \sigma_{33}
    \end{pmatrix}\,,
    \]
   for all ququart input states $\sigma = \sum_{i,j=0}^3 \sigma_{ij}\,|i\rangle\langle j|$.
      \item[(3).] $\{E_0,E_1\}$ and $\{F_0,F_1\}$ can generate a GDS channel $\B_3$ in the expression as
    \[
    \B_3(\sigma) = \begin{pmatrix}
        \sigma_{00} & & \frac{\sigma_{02}}{\sqrt{2}} & \frac{\sigma_{03}}{\sqrt{2}} \\
        & \sigma_{11} & \frac{\sigma_{12}}{\sqrt{2}} & -\frac{\sigma_{13}}{\sqrt{2}} \\
        \frac{\sigma_{20}}{\sqrt{2}} & \frac{\sigma_{21}}{\sqrt{2}} & \sigma_{22} & \\
        \frac{\sigma_{30}}{\sqrt{2}} & -\frac{\sigma_{31}}{\sqrt{2}} & & \sigma_{33}
    \end{pmatrix}\,,
    \]
    for all ququart input states $\sigma = \sum_{i,j=0}^3 \sigma_{ij}\,|i\rangle\langle j|$.
    \end{itemize}
    It's clear that all these GDS channels are different, while $\B_2 = U\B_1 U^\dagger$ with the unitary $U = \begin{pmatrix}
        1 & & & \\
        & -1 & & \\
        & & 1 & \\
        & & & 1
    \end{pmatrix}\,,$ but $\B_3 \neq V\B_1 V^\dagger$ for any unitary $V$. Since the completely dephasing channel is degradable, the GDS channel $\B$ is also degradable, as we proved. Therefore, this example demonstrates that the structure of degradable channels is richer than we anticipated.
\end{ex}

\section{II. Estimations of Quantum Channel Capacities }

In this section, we present general results on the channel capacity of a GDS channel $\N$, including the coherent information $\Q(\N)$, lower bounds on the private information $P^{(1)}(\N)$ and Holevo information $C^{(1)}(\N)$, as well as upper bounds on quantum capacity $Q(\N)$.

\subsection{A. Coherent Information of a GDS channel}\label{Q1Bd}
The coherent information of a quantum channel $\N$ is defined as
\begin{equation}
  \Q(\N) = \max_{\rho} I_c(\rho,\N) \,,
\end{equation}
where $I_c(\rho,\N) := S\big(\N(\rho)\big) - S\big(\N^c(\rho)\big)$, and $S(\rho) = -\tr(\rho \log \rho)$ denotes the von Neumann entropy. Generally, coherent information is hard to evaluate due to the optimization over all possible input states. However, for a GDS channel $\N$, by utilizing the direct sum structure of the whole system, it is sufficient to consider the input state $\rho$ as a direct sum of states $\rho_i$ on the subspaces $\mathcal{H}_{A_i}$.

\begin{prop}
    Let $\N$ be a GDS channel generated by subchannels $\{\N_i\}_{i=0}^n$ with Kraus operators $\{E_k^{(i)}\}$. The coherent information of $\N$ is attained on the block-diagonal state.
\end{prop}
\begin{proof}
Since the output state of the complementary channel for a GDS channel, as given in Eq.~\eqref{complementary}, depends only on the states within the subspaces $\mathcal{H}_{A_i}$. For a given input state $\rho = \begin{pmatrix}
  p_0\rho_0 & C_{01} & \cdots & C_{0n}\\
  C_{01}^\dagger & p_1\rho_1 & \cdots & C_{1n} \\
  \vdots & \vdots & \ddots & \vdots\\
  C_{0n}^\dagger & C_{1n}^\dagger & \cdots & p_n\rho_n
\end{pmatrix}$, the output state $\N^c(\rho)$ is same as $\N^c(\tilde{\rho})$ with $\tilde{\rho} = \begin{pmatrix}
   p_0\rho_0 & &  & \\
   & p_1\rho_1 &  &  \\
   &  & \ddots & \\
   &  &  & p_n\rho_n
\end{pmatrix}$ is the block-diagonal part of $\rho$. 

Next, we introduce some basic knowledge about majorization theory. Recall that for two real vectors $\bds{x} = (x_1,\cdots,x_n)$ and $\bds{y} = (y_1,\cdots,y_n)$ satisfying $\sum_{i=1}^n x_i = \sum_{i=1}^n y_i$, we call $\bds{x}$ is majorized by $\bds{y}$, denoted as $\bds{x} \prec \bds{y}$, if $\sum_{i=1}^k x_{(i)} \leq \sum_{i=1}^k y_{(i)}$ for all $k$ where $x_{(i)}$ denotes $i$-th largest entry of $\bds{x}$. The majorization relation is a partial order on the set of real vectors, and it can be used to compare the entropies of two states. In particular, if $\bds{x} \prec \bds{y}$, then $H(\bds{x}) \geq H(\bds{y})$, where $H(\bds{x})$ is the Shannon entropy of the vector $\bds{x}$. According to an elementary but useful result in matrix majorization theory~\cite{KyFan1954,Marshall1980}:
$$\Big(p_0\bds{\lambda}(\rho_0),\, \cdots, p_n \bds{\lambda}(\rho_n)\Big) = \bds{\lambda}(\tilde{\rho}) \prec \bds{\lambda}(\rho)\,,$$ 
where $\bds{\lambda}(\rho)$ is the eigenvalue vector of $\rho$, we have $S\big(\N(\rho)\big) \leq S\big(\N(\tilde{\rho})\big)$ since the Shannon entropy for the eigenvalue vector of a quantum state equals its von Neumann entropy. Therefore $I_c(\rho,\N) \leq I_c(\tilde{\rho},\N)$, which implies the coherent information of a GDS channel $\N$ is attained on the block-diagonal state $\tilde{\rho}$.  

\end{proof}

Now suppose that $\tilde{\rho} = p_0\rho_0 \oplus \cdots \oplus  p_n \rho_n$, where $\rho_i$ is the quantum states on the subsystem $A_i$, and $\bds{p} = (p_0,\cdots,p_n)$ is probability vector. Let $\omega_i = \N^c_i(\rho_i)$, the coherent information of $\N$ satisfies
\begin{equation}\label{Q1N}
  \begin{aligned}
    \Q(\N) & = \max_{\bds{p},\rho_i} I_c(\tilde{\rho},\N) \\
     & = \max_{\bds{p},\rho_i} \bigg[ H(\bds{p}) + \sum_{i=0}^n\, p_iI_c(\rho_i,\N_i) - \chi\big(\{p_i,\omega_i\}\big) \bigg]\,,
  \end{aligned}
\end{equation}
where the non-negative term $\chi\big(\{p_i,\omega_i\}\big) = S\big(\sum_{i=0}^n \, p_i\omega_i \big) - \sum_{i=0}^n\, p_i S(\omega_i)$ is the Holevo information of the ensemble $\{p_i,\omega_i\}$, which quantifying the concavity of von Neumann entropy with respect to the states $\{\omega_i\}$ and the probability vector $\bds{p}$. 

Therefore, the bounds on the coherent information of a GDS channel $\N$ can be derived by analyzing the Holevo information $\chi\big(\{p_i,\omega_i\}\big)$. There are several bounds on the Holevo information have been proposed in~\cite{Audenaert14,BoundsConcavity,Shirokov19}. Each of these bounds performs optimally in different situations. For example, Ref.~\cite{BoundsConcavity} proved a tight bound on Holevo information for ensembles with two elements, i.e., $\{(x,\omega_0),(1-x,\omega_1)\}$. Here we choose the upper bound proved in~\cite{BoundsConcavity,Audenaert14}, which states 
\begin{equation}\label{BdConcavity} - 
  \frac{1}{2}\sum_{i=0}^n \, p_i ||\omega_i - \bar{\omega}||_1^2 \leq \chi\big(\{p_i,\omega_i\}\big) \leq \frac{H(\bds{p})}{2} \max_{i,j} ||\omega_i - \omega_j||_1 \,,
\end{equation}
where $||A||_1 = \tr \sqrt{A^\dagger A}$ is the trace norm, and $\bar{\omega} = \sum_i\, p_i\omega_i$ is the mixture of the ensemble $\{p_i,\omega_i\}$. Based on this bound, the following lower bounds on the coherent information of a GDS channel $\N$ hold.
\begin{prop}\label{PropLBdQ1}
  Let $\N$ be a GDS channel with subchannels $\{\N_i\}_{i=0}^n $. Then, the coherent information of $\N$ can be bounded from below as
  \begin{equation}\label{lowerBdQ1}
    \begin{aligned}
      \Q(\N) & \geq \max_{\bds{p},\rho_i} \bigg[ H(\bds{p}) + \sum_{i=0}^n\, p_iI_c(\rho_i,\N_i) - \frac{H(\bds{p})}{2} \max_{i,j} ||\omega_i - \omega_j||_1 \bigg]\\ 
      & \geq \log \Big[\sum_i 2^{\Q(\N_i)}\Big] - \frac{\log (n+1)}{2} \min_{\rho_i: \text{opt}} \max_{i,j} ||\omega_i-\omega_j||_1\,,
    \end{aligned}
  \end{equation}
  where the minimum is taken over all optimal states $\rho_i$ such that $I_c(\rho_i,\N_i) = \Q(\N_i)$, and $\omega_i = \N^c_i(\rho_i)$ are the output states of the complementary channels $\N^c_i$.
\end{prop}
\begin{proof}
  We only need to show the second inequality, as the first one follows directly from Eq.~\eqref{BdConcavity}. Note that  
  \begin{align*}
     \max_{\bds{p},\rho_i} \Big[ H(\bds{p}) + \sum_{i=0}^n\, p_iI_c(\rho_i,\N_i)\Big] = \max_{\bds{p}} \Big[ H(\bds{p}) + \sum_{i=0}^n\, p_i \Q(\N_i)\Big] \,,
  \end{align*}
  where $\bds{p}$ denotes a probability distribution, and $\rho_i$ are input states. This expression can be treated as a function of $\bds{p}$, and its maximum can be solved via the Lagrange multiplier method. After some direct computation, the maximum value of the above optimization is $\log \sum_i 2^{\Q(\N_i)}$. Since $H(\bds{p}) \leq \log (n+1)$, we can always choose each $\rho_i$ to be optimal and $\widehat{\omega_i} = \N_i^c(\rho_i)$ so that  
  $$\max_{i,j} ||\widehat{\omega_i}-\widehat{\omega_j}||_1 = \min_{\rho_i: \text{opt}} \max_{i,j} ||\omega_i-\omega_j||_1.$$
\end{proof}

The coherent information of the GDS channel $\N$ generated by subchannels $\{\N_i\}$ has a trivial lower bound as 
\begin{equation}\label{trivialLBdQ1}
  \Q(\N) \geq \max_i\, \Q(\N_i) \,,
\end{equation}
which can be obtained by setting $p_i = 0$ or $1$ in the first inequality in Eq.~\eqref{lowerBdQ1}. It's worth noting that the lower bound in Eq.~\eqref{lowerBdQ1} can be tighter than Eq.~\eqref{trivialLBdQ1} for some special cases. For example, if all $\N_i$ have the same coherent information, i.e., $\Q(\N_i) = \Q(\N_0)$, the lower bound in Eq.~\eqref{lowerBdQ1} gives 
\[
\Q(\N) \geq \Q(\N_0) + \log(n+1) -\frac{\log(n+1)}{2} \min_{\rho_i: \text{opt}} \max_{i,j} ||\omega_i-\omega_j||_1  \geq \Q(\N_0) \,,
\]
while Eq.~\eqref{trivialLBdQ1} only gives the trivial lower bound $\Q(\N) \geq \Q(\N_0) \geq 0$. Furthermore, the lower bound in Eq.~\eqref{lowerBdQ1} implies that the vanishing of coherent information for a GDS channel $\N$ requires not only the vanishing of coherent information for subchannels $\N_i$, but also the trace distance between the output states $\omega_i$ and $\omega_j$ can be as large as $2$, i.e., $\omega_i$ is orthogonal to $\omega_j$ for all $i\neq j$. Therefore, we can choose special subchannels $\{\N_i\}$ such that $Q(\N_i) = 0$ but the trace distance between the output states $\N_i^c(\rho_i)$ and $\N_j^c(\rho_j)$ is always less than $2$, to produce a GDS channel $\N$ with positive coherent information as well as quantum capacity.

On the other hand, the coherent information $\Q(\N)$ has a trivial upper bound as 
\begin{equation}\label{trivialUBdQ1}
  \Q(\N) \leq  \log \Big[\sum_i 2^{\Q(\N_i)}\Big]\,,
\end{equation}
which can be obtained by ignoring the Holevo information term $\chi\big(\{p_i,\omega_i\}\big)$ in the expression of $\Q(\N)$ as in Eq.~\eqref{Q1N}. This upper bound is tight when the GDS channel is generated by two subchannels, with $\N_1$ being the identity channel and $\N_0$ having zero quantum capacity, admitting a pure fixed point~\cite{PCDSC}. In general, the upper bound can be improved by analyzing the Holevo information term $\chi\big(\{p_i,\omega_i\}\big)$. For example, when these subchannels have the same coherent information $\Q(\N_i) = \Q(\N_0)$, using Eq.~\eqref{Q1N} and \eqref{BdConcavity}, we can obtain the following upper bound on the coherent information of a GDS channel $\N$.

\begin{prop}\label{PropUBQ1}
  Let $\N$ be a GDS channel with subchannels $\{\N_i\}_{i=0}^n$. Suppose that all subchannels have the same coherent information $\Q(\N_i) = \Q(\N_0)$. The coherent information of $\N$ can be bounded above as
  \begin{equation}\label{upperBdQ1}
    \begin{aligned}
      \Q(\N) & \leq \Q(\N_0) + \log \Big[ \sum_{i=0}^n \, 2^{-\frac{\min_{\rho_i:\text{opt}} || \omega_i - \bar{\omega}||_1^2}{2}}\Big] \\
    & \leq \log(n+1) + \Q(\N_0) = \log \Big[\sum_{i=0}^n 2^{\Q(\N_i)}\Big]  \,.
    \end{aligned}
  \end{equation}
\end{prop}
\begin{proof}
  According to Eq.~\eqref{BdConcavity}, the coherent information of the GDS channel $\N$ satisfies 
  \begin{align*}
    \Q(\N) & \leq  \max_{\bds{p},\rho_i} \bigg[ H(\bds{p}) + \sum_{i=0}^n\, p_iI_c(\rho_i,\N_i) - \frac{1}{2}\sum_{i=0}^n\, p_i ||\omega_i - \bar{\omega}||_1^2 \bigg] \\
    & \leq \max_{\bds{p}} \bigg[ H(\bds{p}) + \sum_{i=0}^n\, p_i\Q(\N_i) -\frac{1}{2}\sum_{i=0}^n\, p_i \min_{\rho_i:\text{opt}} ||\omega_i - \bar{\omega}||_1^2 \bigg]\,.
  \end{align*}
  The maximum value over $\bds{p}$ can be obtained using the Lagrange multiplier method, giving
  \begin{align*}
    \Q(\N) & \leq \Q(\N_1) + \log \Big[ \sum_{i=0}^n \, 2^{-\frac{\min_{\rho_i:\text{opt}} || \omega_i - \bar{\omega}||_1^2}{2}}\Big] \\
    & \leq \log(n+1) + \Q(\N_1) = \log \Big[\sum_{i=0}^n 2^{\Q(\N_i)}\Big]\,,
    \end{align*}
    where the second inequality is for $\min_{\rho_i:\text{opt}} || \omega_i - \bar{\omega}||_1^2 \geq 0$ and the equality holds only when all $\omega_i = \bar{\omega}$.
\end{proof}

\begin{ex}
   Let two subchannels be the qubit phase-flip $\mathcal{Z}_p(\rho) = (1-p)\rho + p\, \sigma_Z\rho\sigma_Z$ and the bit-flip channels $\mathcal{X}_p(\rho) = (1-p)\rho + p\, \sigma_X\rho\sigma_X$, where $\sigma_X$ and $\sigma_Z$ are Pauli matrices, with their Kraus operators given by $\{\sqrt{1-p}\mathds{1}\,, \sqrt{p}\sigma_Z\}$ and $\{\sqrt{1-p}\mathds{1}\,, \sqrt{p}\sigma_X\}$ respectively. The GDS channel generated by these two subchannels is given by
    \[
    \N(\omega) = (1-p)\omega + p \begin{pmatrix}
        \omega_{00} & -\omega_{01} & \omega_{03} & \omega_{02} \\
        -\omega_{10} & \omega_{11} &  -\omega_{13} & -\omega_{12} \\
       \omega_{30} & -\omega_{31} & \omega_{33} & \omega_{32} \\
       \omega_{20} & -\omega_{21} & \omega_{23} & \omega_{22}
    \end{pmatrix}\,,
    \]
    for any input state $\omega = \sum_{i,j=0}^3 \omega_{ij} |i\rangle\langle j|$. Since both the phase-flip and bit-flip channels are degradable~\cite{Structure}, the GDS channel $\N$ is also degradable. Therefore, its quantum capacity has a single-letter formula, $Q(\N) = \Q(\N)$. According to the optimal states of both $\mathcal{Z}_p$ and $\mathcal{X}_p$ are the completely mixed state $\mathds{1}/2$, and $\mathcal{Z}_p^c(\frac{\mathds{1}}{2}) = \mathcal{X}_p^c(\frac{\mathds{1}}{2})$, we obtain 
    \begin{equation}
        Q(\N) = \Q(\N) = 2 - H_b(p)\,,
    \end{equation}
    where $H_b(p) = -p\log p - (1-p)\log(1-p)$ is a binary entropy function. The quantum capacity of the GDS channel $\N$ is strictly greater than that of the channel obtained by the coherent control of the same subchannels, which has the lower bound $p - H_b(p) + H_b(\frac{1-p}{2})$~\cite{CoheControl}. Even though the input-output dimension of the GDS channel is four, dividing the channel capacity by $2$ to get $1-H_b(p)/2$, the resulting value remains larger than the quantum capacity generated by coherent control. This result reveals that our GDS construction possesses communication advantages compared to coherent control of subchannels~\cite{CoheControl} and quantum switch~\cite{salek2018}.
\end{ex}

\subsection{B. Lower bounds on private and Holevo information for a GDS channel}\label{P1C1Bd}

While the quantum capacity measures the maximum rate for asymptotically transmitting quantum information with vanishing errors through a quantum channel $\B$, the private $P$ and classical $C$ capacities portray the highest rates for the quantum channel $\B$ transmitting private and classical information with asymptotically vanishing errors. The private and classical capacities for a quantum channel $\B$ are given by the regularized expressions~\cite{LSD3,RegExpP} and~\cite{HSW1,HSW2}
\begin{equation*}
    P(\B) = \lim_{n\rightarrow \infty} \frac{P^{(1)}(\B^{\otimes n})}{n}\,, \quad C(\B) = \lim_{n\rightarrow \infty} \frac{C^{(1)}(\B^{\otimes n})}{n}\,.
\end{equation*}
The private information $P^{(1)}$ and Holevo information $C^{(1)}$ are defined by taking the maximum of $I_p$ and $\chi$ over all possible ensembles $\mathfrak{E} = \{p_i,\rho_i\}$, respectively; that is
\begin{equation*}
    P^{(1)} = \max_{\{p_i,\rho_i\}} I_p\Big(\{p_i,\rho_i\}, \B\Big)\,, \quad  C^{(1)} = \max_{\{p_i,\rho_i\}} \chi\Big(\{p_i,\rho_i\}, \B\Big)\,,
\end{equation*}
where $I_p$ and $\chi$ of the channel $\B$ with respect to an ensemble $\{p_i,\rho_i\}$ are defined as 
\begin{equation*}
    \begin{aligned}
        I_p\Big(\{p_i,\rho_i\}, \B\Big) & = I_c\Big(\sum_i p_i\rho_i\,, \B\Big) - \sum_i p_i I_c(\rho_i,\B)\,, \\
        \chi\Big(\{p_i,\rho_i\}, \B\Big) & = S\Big(\B\big(\sum_i p_i\rho_i\big)\Big) - \sum_i p_i S\big(\B(\rho_i)\big)\,.
    \end{aligned}
\end{equation*}

Inspired by the optimal state achieving the coherent information of a GDS channel $\N$ is block-diagonal, we now use the ensemble consisting of block-diagonal states to derive lower bounds on the private information $P^{(1)}(\N)$ and the Holevo information $C^{(1)}(\N)$ for a GDS channel $\N$. 
\begin{prop}\label{PropLBdP1}
  Let $\N$ be a GDS channel generated by subchannels $\{\N_i\}_{i=0}^n$. The private information $P^{(1)}(\N)$ of $\N$ can be bounded below as
  \begin{equation}\label{lowerBdP1}
   \begin{aligned}
       P^{(1)}(\N) & \geq \max_{\bds{p},\mathfrak{E}'_i} \bigg[H(\bds{p}) + \sum_{i=0}^n p_i I_p\big(\mathfrak{E}'_i,\N_i\big)  - \chi\big(\{p_i,\omega_i\}\big) \bigg] \\
      & \geq \log \Big[\sum_{i=0}^n 2^{P^{(1)}(\N_i)}\Big] - \frac{\log (n+1)}{2} \min_{\rho_i: \text{opt}} \max_{i,j} ||\omega_i-\omega_j||_1\,,
    \end{aligned}
  \end{equation}
  where $\mathfrak{E}'_i = \{r_j^{(i)},\rho_j^{(i)}\}$ is an ensemble for the subchannel $\N_i$, and $\omega_i = \N_i^c(\sum r_j^{(i)}\rho_{j}^{(i)})$. The minimum is taken over all optimal ensembles $\{\mathfrak{E}'_i\}$ such that $I_p(\mathfrak{E}'_i,\N_i) = P^{(1)}(\N_i)$, and $\rho_i = \sum_j r_j^{(i)}\rho_j^{(i)}$ is the mixture state of the optimal ensemble.
\end{prop}
\begin{proof}
  Let $\mathfrak{E}_i = \{r_j^{(i)}, \rho_j^{(i)} \oplus 0_{A_i^c}\}$ be an ensemble with quantum states $\rho_j^{(i)}$ on subsystem $A_i$, and $0_{A_i^c}$ is the zero operator acting on all subsystems except $A_i$. Define a new ensemble $\mathfrak{E} = \{p_i,\mathfrak{E}_i\}$ to be the probable union of these ensembles $\{\mathfrak{E}_i\}$. Let $\rho_i = \sum_j\, r_j^{(i)} \rho_j^{(i)}$, $\omega_i = \N_i^c(\rho_i)$ and $\rho = \bigoplus_{i=0}^n p_i \rho_i$, we have 
\begin{align*}
    I_p\Big(\mathfrak{E},\N \Big) & = I_c(\rho,\N) - \sum_{i=0}^n p_i \sum_j r_j^{(i)} I_c(\rho_j^{(i)},\N_i) \\
    & = H(\bds{p}) + \sum_{i=0}^n p_i I_p\big(\mathfrak{E}'_i,\N_i\big)  - \chi\big(\{p_i,\omega_i\}\big)\,,
\end{align*}
where $\mathfrak{E}'_i = \{r_j^{(i)}, \rho_j^{(i)} \}$ is the induced ensemble for the subchannel $\N_i$. Thus, we have 
  \begin{equation}
    \begin{aligned}
       P^{(1)}(\N) & \geq \max_{\bds{p},\mathfrak{E}'_i} \bigg[H(\bds{p}) + \sum_{i=0}^n p_i I_p\big(\mathfrak{E}'_i,\N_i\big)  - \chi\big(\{p_i,\omega_i\}\big) \bigg] \\
     & \geq  \max_{\bds{p},\mathfrak{E}'_i} \bigg[H(\bds{p}) + \sum_{i=0}^n p_i I_p\big(\mathfrak{E}'_i,\N_i\big) - \frac{H(\bds{p})}{2} \max_{i,j} ||\omega_i - \omega_j||_1 \bigg]\\ 
      & \geq \log \Big[\sum_i 2^{P^{(1)}(\N_i)}\Big] - \frac{\log (n+1)}{2} \min_{\rho_i: \text{opt}} \max_{i,j} ||\omega_i-\omega_j||_1\,.
    \end{aligned}
  \end{equation}

\end{proof}

\begin{prop}\label{PropLBdC1}
 Let $\N$ be a GDS channel generated by subchannels $\{\N_i\}_{i=0}^n$. The Holevo information $C^{(1)}(\N)$ of $\N$ can be bounded below as
  \begin{equation}\label{lowerBdC1}
    C^{(1)}(\N) \geq \log \Big[\sum_{i=0}^n 2^{C^{(1)}(\N_i)}\Big] \,.
  \end{equation}
\end{prop}
\begin{proof}
 Let $\mathfrak{E}_i = \{r_j^{(i)}, \rho_j^{(i)} \oplus 0_{A_i^c}\}$ be an ensemble with quantum states $\rho_j^{(i)}$ on subsystem $A_i$, and $0_{A_i^c}$ is the zero operator acting on all subsystems except $A_i$. Define a new ensemble $\mathfrak{E} = \{p_i,\mathfrak{E}_i\}$ to be the probable union of these ensembles $\{\mathfrak{E}_i\}$. Let $\rho_i = \sum_j\, r_j^{(i)} \rho_j^{(i)}$ and $\rho = \bigoplus_{i=0}^n p_i \rho_i$, we have
\begin{align*}
    \chi\Big(\mathfrak{E},\N\Big) & = S\big[\N(\rho)\big] - \sum_{i=0}^n p_i \sum_j r_j^{(i)} \, S\big[\N_i(\rho_j^{(i)})\big] \\
    & = H(\bds{p}) + \sum_{i=0}^n p_i\, \chi(\mathfrak{E}'_i,\N_i)\,,
  \end{align*}
  where $\mathfrak{E}'_i = \{r_j^{(i)}, \rho_j^{(i)} \}$ is the induced ensemble for the subchannel $\N_i$. Thus, we have 
  \begin{align*}
    C^{(1)}(\N) & \geq \max_{\bds{p},\mathfrak{E}'_i} \bigg[H(\bds{p}) + \sum_{i=0}^n p_i\, \chi(\mathfrak{E}'_i,\N_i) \bigg] \\
    & = \log \Big[\sum_{i=0}^n 2^{C^{(1)}(\N_i)}\Big] \,.
  \end{align*}
\end{proof}

It remains unclear whether the ensemble consisting of block-diagonal states attains the exact values of the private and Holevo information for a generic GDS channel. However, as shown in the last section, for a special class of GDS channels generated by completely depolarizing subchannels, the ensemble consisting of block-diagonal states indeed achieves both the private and Holevo informations. Furthermore, for these special GDS channels, the private capacity equals the private information, and the classical capacity equals the Holevo information. 

On the other hand, since one can take the block dephasing channel 
$$\Delta(O) = \begin{pmatrix}
O_{00} & & \\
& \ddots & \\
& & O_{nn}
\end{pmatrix}\,,$$ 
such that $\Delta\circ \N = \N_0\oplus\cdots \oplus \N_n$ for a GDS channel $\N$, where $\N_0\oplus\cdots \oplus \N_n$ is the direct sum of subchannels $\{\N_i\}$. The coherent, private, and Holevo information for the direct sum channel $\N_0\oplus\cdots \oplus \N_n$ have been derived in Refs.~\cite{DirectSumC,UnboundP} as below.
\begin{prop}[\cite{DirectSumC,UnboundP}]
  Let $\{\N_i\}$ be a class of quantum channels. The coherent, private, and Holevo information for the direct sum channel $\N_0\oplus\cdots \oplus \N_n$ are given by
  \begin{align*}
    \Q(\N_0\oplus\cdots \oplus \N_n) & =\max_i\, \Q(\N_i)\,, \\
    P^{(1)}(\N_0\oplus\cdots \oplus \N_n) & = \max_i\, P^{(1)}(\N_i) \,, \\   
    C^{(1)}(\N_0\oplus\cdots \oplus \N_n) & = \log \Big[\sum_{i=0}^n 2^{C^{(1)}(\N_i)}\Big]\,.
  \end{align*}
\end{prop}
Therefore, the data processing inequality gives rise to that the coherent, private and Holevo information of a GDS channel $\N$ generated by subchannel $\{\N_i\}$ can be bounded below by $\Q,P^{(1)},C^{(1)}$ of the direct sum channel $\N_0\oplus\cdots \oplus \N_n$, respectively. Although Proposition~\ref{PropLBdC1} does not provide a better lower bound on the Holevo information of a GDS channel $\N$ than $C^{(1)}(\N_0\oplus\cdots \oplus \N_n)$, the lower bounds on coherent and private information provided in Propositions~\ref{PropLBdQ1} and \ref{PropLBdP1} can be tighter than the coherent and private information of $\N_0\oplus\cdots \oplus \N_n$. In particular, when all subchannels $\N_i$ have same coherent information $\Q(\N_0) = \Q(\N_i)$, the discussion after Eq.~\eqref{trivialLBdQ1} has already mentioned that Proposition~\ref{PropLBdQ1} provides a tighter lower bound. Especially, Eq.~\eqref{lowerBdP1} is tighter than $P^{(1)}(\N_0\oplus\cdots \oplus \N_n)$ when the private informations of all $\N_i$ are equal.   

\subsection{C. Upper bounds on quantum capacity for a GDS channel}\label{QBd}

Quantum capacity quantifies the maximum rate of quantum information that can be transmitted through a quantum channel with arbitrarily low probability of error. However, due to quantum effects, the mathematical characterization of quantum capacity is more involved than the classical channel capacity, as it is given by the following regularized expression~\cite{LSD1,LSD2,LSD3}
\begin{equation}\label{regexpSM}
    Q(\B) = \lim_{n\rightarrow \infty} \frac{Q^{(1)}(\B^{\otimes n})}{n}\,. 
\end{equation}

When channel $\B$ is degradable, the quantum capacity is equal to the coherent information $Q(\B) = \Q(\B)$. According to Theorem~\ref{degradablity}, the GDS channel $\N$ is degradable if and only if all subchannels $\N_i$ are degradable. In this case, the quantum capacity of a GDS channel $\N$ can be computed by the expression in Eq.~\eqref{Q1N}. 

For a non-degradable channel, however, the quantum capacity cannot be equal to the coherent information but is instead given by the regularized expression in Eq.~\eqref{regexpSM}. In this section, we will derive an analytic upper bound on the quantum capacity for a generic GDS channel $\N$. The tool we use is the well-known ``transposition bound''~\cite{TranspositionBound}, which states that the quantum capacity of a quantum channel $\B: A' \to B$ is upper bounded by $Q(\B) \leq \log ||\mathcal{T}\circ \B||_{\diamond}$, where $\mathcal{T}$ denotes the transposition map, and $||\cdot||_{\diamond}$ is the diamond norm, defined for a linear map $\B$ as  
\[
||\B||_{\diamond} = \sup_{X_{AA'}} \{||\mathcal{I}_{A}\otimes \B(X)||_1\,,\, ||X||_1 \leq 1\} \,.
\] 
In particular, the transposition bound can be represented as the following SDP~\cite{Fang2021}
\begin{equation}\label{TransBdSM}
  ||\mathcal{T}\circ \B||_{\diamond} = \min \{y| Y_{AB}\pm \mathscr{C}_{\B}^{T_B} \geq 0\,,\, \tr_B Y_{AB} \leq y\mathds{1}\}\,,
\end{equation}
where $\mathscr{C}_{\B} = \mathcal{I}_A \otimes \B(|\Phi\rangle\langle\Phi|)$ is the Choi matrix of the channel $\B$ with the non-normed maximum entangled state $|\Phi\rangle = \sum_{i} |ii\rangle_{AA'}$, and $T_B$ is the partial transposition on the system $B$.

We now provide an upper bound on the quantum capacity for a GDS channel $\N$ in terms of the transposition bound on the quantum capacity for subchannels $\N_i$. In particular, denotes the Choi matrix for the linear map $\M_{ij}$ in the expression of the GDS channel $\N$ as
\begin{equation}\label{ChoiMSM}
\mathscr{C}_{\M_{ij}} = \sum_{s=0}^{d_{A_i}-1}\sum_{t=0}^{d_{A_j}-1} |s\rangle\langle t|\otimes \sum_k E_{k}^{(i)}|s\rangle\langle t|E_{k}^{(j),\dagger}\,,
\end{equation}
where $\{|s\rangle\}_{i=0}^{d_{A_i}-1}$ and $\{|t\rangle\}_{j=0}^{d_{A_j}-1}$ are the orthonormal bases of the subspaces $\mathcal{H}_{A_i}$ and $\mathcal{H}_{A_j}$ respectively. Moreover, we have $\mathscr{C}_{\M_{ij}}^\dagger = \mathscr{C}_{\M_{ji}}$. Since $d_{A_i}$ may not be equal to $d_{A_j}$, the map $\M_{ij}$ does not act on the square matrices, and its Choi matrix $\mathscr{C}_{\M_{ij}}$ can  not be Hermitian. Recall the matrix absolute value $|M| = \sqrt{M^\dagger M}$ is defined for any matrix $M$, we have following Lemma.
\begin{lem}\label{Lem1}
  Let $M$ be a matrix with its absolute value operator $|M|$. For any two positive numbers $a$ and $b$, we have 
  \begin{equation}
    \pm\Big(M + M^\dagger \Big) \leq \frac{b}{a} |M| + \frac{a}{b} |M^\dagger|\,.
  \end{equation}
\end{lem}
\begin{proof}
 The matrix $M$ and its Hermitian conjugate $M^\dagger$ have singular value decompositions
  \[
  M = \sum_l s_l \bds{x}_l \bds{y}_l^\dagger\,, \quad M^\dagger = \sum_l s_l \bds{y}_l \bds{x}_l^\dagger\,,
  \]
  where $s_l$ are the singular values of $M$, and $\bds{x}_l$  ($\bds{y}_l$) are the left (right) singular vectors. Then the absolute value operators of $M$ and $M^\dagger$ are 
  \[
  |M| = \sum_l s_l \bds{y}_l \bds{y}_l^\dagger\,,\quad |M^\dagger| = \sum_l s_l \bds{x}_l \bds{x}_l^\dagger\,.
  \]
  As inequalities
  \[
  (a\bds{x} + b\bds{y})(a\bds{x} + b\bds{y})^\dagger \geq 0\,,\, (a\bds{x} - b\bds{y})(a\bds{x} - b\bds{y})^\dagger \geq 0
  \]
  hold for any two vectors $\bds{x}$, $\bds{y}$, and any two positive numbers $a,b$, we have
  \begin{equation*}
    \pm\Big(M + M^\dagger \Big) \leq \frac{b}{a} |M| + \frac{a}{b} |M^\dagger|\,.
  \end{equation*}
\end{proof}
We are now ready to prove the upper bound on the quantum capacity for a GDS channel.
\begin{thm}\label{UpBdQNSM}
  Let $\N$ be a GDS channel generated by subchannels $\{\N_i\}_{i=0}^n$ with Kraus operators $\{E_k^{(i)}\}$. The quantum capacity of $\N$ can be upper bounded as
  \begin{equation}\label{upperBdQSM}
    \begin{aligned}
      Q(\N) & \leq \log \max_i \Big[||\mathcal{T}\circ \N_i||_{\diamond}  + \sum_{j\neq i} \sqrt{||\tr_B(|\mathscr{C}_{\M_{ij}}^{T_B}|)||_\infty\,||\tr_B(|\mathscr{C}_{\M_{ji}}^{T_B}|)||_\infty}\Big] \\
    &\leq \log \max_i \Big[||\mathcal{T}\circ \N_i||_{\diamond} + \sum_{j\neq i} ||\mathscr{C}_{\M_{ij}}^{T_B}||_1\Big] \,,
    \end{aligned}
  \end{equation}
  where $||\cdot||_\infty$ is the spectral norm and $||\cdot||_1$ is the trace norm.
\end{thm}
\begin{proof}
Suppose that $(y_i,Y_{A_iB_i})$ is the optimal solution of the SDP in Eq.~\eqref{TransBdSM} for subchannels $\N_i$, for all $i = 0,\cdots,n$. The Choi matrix $\mathscr{C}_{\N}$ for a GDS channel $\N$ is given by
  \begin{align*}
    \mathscr{C}_{\N} & = \sum_{i=0}^n \mathscr{C}_{\N_i} + \sum_{i<j} \Big(\mathscr{C}_{\M_{ij}} + \mathscr{C}_{\M_{ji}}\Big) \,.
  \end{align*}
  Let the Hermitian matrix $Y_{AB}$ be 
  \begin{equation}\label{YAB}
      Y_{AB} = \sum_{i=0}^n Y_{A_iB_i} + \sum_{i<j} \big( \frac{b_{ij}}{a_{ij}} |\mathscr{C}_{\M_{ij}}^{T_B}| + \frac{a_{ij}}{b_{ij}}|\mathscr{C}_{\M_{ij}^\dagger}^{T_B}|\big)\,,
  \end{equation}
  where $a_{ij}$ and $b_{ij}$ are positive real numbers to be determined. Since the matrix $Y_{AB}$ is required to satisfy the condition $Y_{AB} \pm \mathscr{C}_{\N}^{T_B} \geq 0$, and all operators $\{\mathscr{C}_{\N_i}^{T_B}, \mathscr{C}_{\M_{ij}}^{T_B}\}$ are mutually orthogonal, it is sufficient to show that 
  \begin{equation}
    \frac{b_{ij}}{a_{ij}}|\mathscr{C}_{\M_{ij}}^{T_B}| + \frac{a_{ij}}{b_{ij}}|\mathscr{C}_{\M_{ij}^\dagger}^{T_B}| \geq \pm \Big(\mathscr{C}_{\M_{ij}}^{T_B} + \mathscr{C}_{\M_{ji}}^{T_B} \Big)\,,
  \end{equation}
  which follows from Lemma~\ref{Lem1}. 
  
  Since $\tr_B(Y_{A_iB_i}) \leq y_i\mathds{1}_{A_i}$ and 
  \begin{align*}
  \tr_B\Big(\frac{b_{ij}}{a_{ij}} |\mathscr{C}_{\M_{ij}}^{T_B}|\Big) \leq \frac{b_{ij}}{a_{ij}} ||\tr_B(|\mathscr{C}_{\M_{ij}}^{T_B}|)||_\infty \mathds{1}_{A_i}\,, \qquad \tr_B\Big(\frac{a_{ij}}{b_{ij}} |\mathscr{C}_{\M_{ij}^\dagger}^{T_B}|\Big) \leq \frac{a_{ij}}{b_{ij}} ||\tr_B(|\mathscr{C}_{\M_{ij}^\dagger}^{T_B}|)||_\infty \mathds{1}_{A_j}\,,
  \end{align*}
  we can choose positive numbers $a_{ij} = \sqrt{||\tr_B(|\mathscr{C}_{\M_{ij}}^{T_B}|)||_\infty}$ and $b_{ij} = \sqrt{||\tr_B(|\mathscr{C}_{\M_{ij}^\dagger}^{T_B}|)||_\infty}$, such that
  \[
  \tr_B\Big(\frac{b_{ij}}{a_{ij}} |\mathscr{C}_{\M_{ij}}^{T_B}| + \frac{a_{ij}}{b_{ij}} |\mathscr{C}_{\M_{ij}^\dagger}^{T_B}|\Big) \leq a_{ij}b_{ij}\, (\mathds{1}_{A_i}\oplus \mathds{1}_{A_j})\,.
  \]  
  Then we have
  \[
  \tr_B(Y_{AB}) \leq \max_i \Big[y_i  + \sum_{j\neq i} \sqrt{||\tr_B(|\mathscr{C}_{\M_{ij}}^{T_B}|)||_\infty\,||\tr_B(|\mathscr{C}_{\M_{ji}}^{T_B}|)||_\infty}\Big]\,.
  \]
  Because of $y_i = ||\mathcal{T}\circ \N_i||_{\diamond}$, we can obtain the upper bound on the quantum capacity for a GDS channel $\N$ as the first inequality. 

  Finally, given a bipartite operator $X_{AB}$ with dimensions $d_A$ and $d_B$, and $X_A = \tr_B(X_{AB})$. Define the norm $||X_{AB}||_{(k)} = \sum_{l=1}^k s_l(X_{AB})$ where $s_l(X_{AB})$ is the $l$-th largest singular value of $X_{AB}$. It holds~\cite{Rastegin2012}
  \[
  ||X_{A}||_{\infty} \leq ||X_{AB}||_{(d_B)}  \,.
  \]
  Since $||X_{AB}||_{(d_B)} \leq ||X_{AB}||_1$, we have 
  \begin{align*}
    ||\tr_B(|\mathscr{C}_{\M}^{T_B}|)||_\infty  \leq ||\mathscr{C}_{\M}^{T_B}||_1 \,, \quad 
    ||\tr_B(|\mathscr{C}_{\M^\dagger}^{T_B}|)||_\infty  \leq ||\mathscr{C}_{\M^\dagger}^{T_B}||_1\,.
  \end{align*}
 While the trace norm for $\mathscr{C}_{\M}^{T_B}$ and $\mathscr{C}_{\M^\dagger}^{T_B}$ are equal according to their singular value decompositions, we can obtain the second inequality in Eq.~\eqref{upperBdQSM}.
\end{proof}
While Wang et al.~\cite{SDPBoundQ} improved the transposition bound on quantum capacity by using the quantity $\Gamma(\B)$, which can be computed by an SDP, it is natural to ask whether the upper bound in Eq.~\eqref{upperBdQSM} can be improved by using the quantity $\Gamma(\N)$ for a GDS channel $\N$. For some special GDS channels generated by completely depolarizing subchannels, including the ``platypus channels'' introduced in~\cite{platypusTIT}, the quantity $\log \Gamma(\N)$ gives the same upper bound as the first inequality in Eq.~\eqref{upperBdQSM}~\cite{platypusTIT}. It is not clear whether the quantity $\Gamma(\N)$ can improve the upper bound in Eq.~\eqref{upperBdQSM} for general GDS channels. Furthermore, Fang and Fawzi~\cite{Fang2021} recently gave a tighter upper bound on the quantum capacity for a generic channel $\B$,
\begin{equation}
  Q(\B) \leq \widehat{R}_{\alpha}(\B) \leq \log \Gamma(\B) \leq \log ||\mathcal{T}\circ \B||_{\diamond} \,,
\end{equation}
where $\widehat{R}_{\alpha}(\B)$ can also be computed by an SDP, but it is too complicated to be used to give an analytic upper bound on the quantum capacity for a GDS channel $\N$.

\section{III. A new class of channels with a single-letter formula for quantum capacity}

The class of channels with a single-letter formula for quantum channel capacities is scarce, especially in terms of quantum capacity, mainly including the regularized less noisy channel whose complement has zero private capacity~\cite{Watanabe,IEEEPartialOrder}, informationally degradable channels~\cite{UniformAdd}, as well as the channel with zero quantum capacity maintaining positive partial transposition (PPT) and antidegradable channels. Here, motivated by the upper and lower bounds on coherent information for GDS channels provided in Propositions~\ref{PropLBdQ1} and~\ref{PropUBQ1}, we give the sufficient condition for a GDS channel $\N$ generated by the subchannels $\{\N_i\}$ with Kraus operators $\{E_k^{(i)}\}$ satisfying $Q(\N) = \Q(\N) = \log \big[ \sum_i 2^{\Q(\N_i)}\big]$.

Suppose that the GDS channel is generated by subchannels $\{\N_i\}_{i=0}^n$ with Kraus operators $\{E_k^{(i)}\}$, the Kraus operator of $\N$ is given by $E_k = \oplus_{i}\, E_k^{(i)}$, according to Theorem~\ref{KrausSM}. Then the Kraus operators of $\N^{\otimes m}$ are 
\[
E_{k_1} \otimes \cdots \otimes E_{k_m} = \big(\bigoplus_{i=0}^n\, E_{k_1}^{(i)}\big) \otimes \cdots \otimes \big(\bigoplus_{i=0}^n\, E_{k_m}^{(i)}\big) \,,
\]
which retain the direct sum structure. Hence, $\N^{\otimes m}$ is also a GDS channel for all $m \in \mathbb{N}$. Indeed, the tensor product channel $\N^{\otimes m}$ can be regarded as a GDS channel generated by subchannels $\{\N_{\bds{i}} = \N_{i_1}\otimes \cdots \otimes \N_{i_m}\}$ for $\bds{i} = (i_1,\cdots,i_m) \in \{0,1,\cdots,n\}^{\times m}$. Therefore, the trivial upper bound on coherent information for a GDS channel provided in Proposition~\ref{PropUBQ1} gives 
\begin{equation}\label{Condi1}
    \Q(\N^{\otimes m}) \leq \log \bigg[\sum_{i_1,\cdots,i_m = 0}^n 2^{\Q(\N_{\bds{i}})}\bigg]\,.
\end{equation}
On the other hand, Proposition~\ref{PropLBdQ1} gives the lower bound on coherent information for $\N^{\otimes m}$,
\begin{equation}
    \Q(\N^{\otimes m}) \geq \log \bigg[\sum_{i_1,\cdots,i_m = 0}^n 2^{\Q(\N_{\bds{i}})}\bigg] - \frac{m\log(n+1)}{2} \min_{\rho_{\bds{i}}: \text{opt}} \max_{\bds{i},\bds{j}} ||\omega_{\bds{i}}-\omega_{\bds{j}}||_1\,,
\end{equation}
where $\rho_{\bds{i}}$ is the optimal state for coherent information of $\N_{\bds{i}}$ and $\omega_{\bds{i}} = \N^c_{\bds{i}}(\rho_{\bds{i}})$ is the output state of the complementary channel $\N^c_{\bds{i}} = \N^c_{i_1}\otimes \cdots \otimes \N^c_{i_m}$ with $\bds{i} = (i_1,\cdots,i_m)$. To establish $Q(\N) = \Q(\N) = \log \big[ \sum_i 2^{\Q(\N_i)}\big]$, we require the following two conditions:
\begin{itemize}
    \item[(1).] The subchannels have strongly additive coherent information when combined with one another, that is,  $\Q(\N_{\bds{i}}) = \sum_{r=1}^m \Q(\N_{i_r})$. Then Eq.~\eqref{Condi1} becomes $\Q(\N^{\otimes m}) \leq m\log \big[\sum_{i=1}^n 2^{\Q(\N_{i})}\big]$ for all $m \in \mathbb{N}$, which implies $Q(\N) \leq \log (\sum_{i=1}^n 2^{\Q(\N_{i})})$.
    \item[(2).] There exist optimal states $\{\rho_{\bds{i}}\}$ such that $\max_{\bds{i},\bds{j}} ||\omega_{\bds{i}}-\omega_{\bds{j}}||_1 = 0$, that is $\omega_{\bds{i}} = \omega_{\bds{j}}$ for all indices $\bds{i}$ and $\bds{j}$. Then, the lower bound on $\Q(\N^{\otimes m})$ matches the upper bound. Totally, it gives $Q(\N) = \Q(\N) =  \log \big[ \sum_i 2^{\Q(\N_i)}\big]$.
\end{itemize}

Now, we aim to identify the class of channels that satisfy these two conditions. For the first condition, it is clear that the degradable channel meets this requirement. Moreover, once all subchannels are degradable, the GDS channel is also degradable and hence is automatically equipped with a single-letter formula for quantum capacity. Therefore, we focus on the special case in which all subchannels are either antidegradable or PPT, which implies that the tensor channel $\N_{\bds{i}}$ is either antidegradable or PPT, and thus has zero quantum capacity. Here, we cannot assume that some subchannels are antidegradable while others are PPT, since combining PPT and antidegradable channels can produce superactivation~\cite{superactivation}, which is an extreme phenomenon of superadditivity in quantum capacity. 

For the second condition, there exist optimal states $\{\rho_{\bds{i}}\}$ such that $\N_{\bds{i}}^c(\rho_{\bds{i}}) = \N_{\bds{j}}^c(\rho_{\bds{j}})$. Since the tensor product channel $\N_{\bds{i}}$ has zero quantum capacity, these optimal states can be chosen as tensor products of pure states, $\rho_{\bds{i}} = |\psi_{i_1}\rangle\langle\psi_{i_1}|\otimes \cdots \otimes |\psi_{i_m}\rangle\langle\psi_{i_m}|$, where $|\psi_{i_r}\rangle\langle\psi_{i_r}|$ is the optimal state of the subchannel $\N_{i_r}$. Then, the second condition becomes, for all $i,j\in\{0,1,\cdots,n\}$, the output states of the complementary channels are identical, $\N_{i}^c(|\psi_i\rangle\langle\psi_i|) = \N^c_j(|\psi_j\rangle\langle\psi_j|)$.

Totally, we can obtain the following theorem for the GDS channel with the single-letter formula for quantum capacity.
\begin{thm}\label{SingleSM}
    Given a class of subchannels $\{\N_i\}_{i=0}^n$ that are all antidegradable or all PPT. If there exist pure states $\{|\psi_i\rangle\langle\psi_i|\}_{i=0}^n$ such that the output states of the complementary channels for these subchannels are identical, $\N_i^c(|\psi_i\rangle\langle\psi_i|) = \N_j^c(|\psi_j\rangle\langle\psi_j|)$ for all $i,j$, then the GDS channel $\N$ generated by these subchannels has a single-letter formula for its quantum capacity
    \begin{equation}
        Q(\N) = \Q(\N) = \log \big[ \sum_i 2^{\Q(\N_i)}\big] = \log (n+1)\,.
    \end{equation}
\end{thm}

Since restricting $\N^c$ to the input subsystem $A_i$ results in $\N^c_i$, which may have positive coherent information, the coherent and private information of $\N^c$ can be positive. Therefore, the GDS channel $\N$ is not less noisy if one of the complementary channels of its subchannels has positive coherent information. Since both regularized less noisy channels and informationally degradable channels are special types of less noisy channels~\cite{AddnonDeg}, our theorem gives rise to a new class of quantum channels with a single-letter formula for quantum capacity. 

\begin{cor}
    Given a quantum channel $\B$ that is either antidegradable or PPT. The GDS channel generated by $n+1$ copies of $\B$ (regardless of the implementation of $\B$) satisfies 
    \begin{equation*}
        Q(\N) = \Q(\N) = \log (n+1).
    \end{equation*}
\end{cor}
One can consider producing a GDS channel generated by only two subchannels: one being the identical channel $\mathcal{I}$, and the other subchannel $\N_1$ having zero quantum capacity. These two subchannels also satisfy the first condition discussed above. Since the complementary channel of identical channel is $\tr$, mapping every input state to a fixed state $|\alpha\rangle\langle\alpha|_E$, the second condition comes to there exists a pure state $|\psi\rangle\langle \psi|$ such that $\N_1^c(|\psi\rangle\langle\psi|) = |\alpha\rangle\langle \alpha|_E$, which revisits the result in~\cite{PCDSC}. We now give some explicit examples of GDS channels with a single-letter formula for quantum capacity.
\begin{ex}
    Given a class of entanglement-breaking subchannels $\{\N_i\}_{i=0}^n$, their complementary channels are Hadamard channels. Suppose each complementary channel $\N^c_i: A_i\to E$ is parameterized by a positive semidefinite matrix $M_i$ whose diagonal entries are all equal to one, $\N_i^c(\rho_i) = M_i \odot \rho_i $ where $\odot$ denotes the Hadamard (entrywise) matrix product. We can choose the optimal state for coherent information of $\N_i$ to be $|0\rangle\langle 0|_{A_i}$, such that $\N_i^c(|0\rangle\langle 0|_{A_i}) = |0\rangle\langle 0|_{E}$, which is fixed for all $i$. Thus, according to Theorem~\ref{SingleSM}, these subchannels can produce a GDS channel with a single-letter formula for its quantum capacity.
\end{ex}

\begin{ex}
    Given a class of antidegradable qubit amplitude damping subchannels $\{\mathcal{A}_{\gamma_i}\}$ with $\gamma_i \geq 1/2$, their complements are also qubit amplitude damping channels with parameters $\{1-\gamma_i\}$. Since $\mathcal{A}_{\gamma_i}^c(|0\rangle\langle 0|) = |0\rangle\langle 0|$ is fixed, these subchannels can thus generate a GDS channel with the single-letter formula for its quantum capacity.
\end{ex}

For a GDS channel satisfying the condition in Theorem~\ref{SingleSM}, according to Proposition~\ref{PropLBdP1}, its private information is at least $\log(n+1)$. However, we do not know whether the GDS channel also has a single-letter formula for private capacity (and whether it equals the quantum capacity), even when considering the special examples discussed above and the general upper bound on private capacity derived in Refs.~\cite{UpBdP,Fang2021}.

\section{IV. completely depolarizing subchannels}

In this section, we consider a special GDS channel generated by a set of different completely depolarizing subchannels $\{\N_i\}$. Each subchannel $\N_i$ acts from the input subsystem $A_i$  to the output subsystem $B_i$, mapping every operator on $A_i$ to the completely mixed state on $B_i$,
\begin{equation}\label{CDCsubchannels}
\N_i: A_i \to B_i,\quad \N_i(\rho_i) = \mathds{1}_{B_i}/d_{B_i}\,,
\end{equation}
with Kraus operators as
\begin{equation*}
    E_{s,t}^{(i)} = \frac{1}{\sqrt{d_{B_i}}}\, |s\rangle\langle t |\,, \quad s = 0,1,\cdots,d_{B_i}-1,\quad t = 0,1,\cdots,d_{A_i}-1\,.
\end{equation*} 
Such channel $\N_i$ cannot transmit any useful information and thus has zero channel capacities. However, when completely depolarizing channels are used to generate a GDS channel $\N$, the GDS channel has positive channel capacities. For simplicity in the computation, we require that $d_{A_i} d_{B_i} = d$ is fixed, and the total dimension for the input system $A$ is $d_A = \sum_i d_{A_i}$. 

\subsection{A. Bounds on coherent information and quantum capacity}

The complementary channels of these subchannels are given by $\omega_i = \N_i^c(\rho_i) = \frac{\mathds{1}_{B_i}}{d_{B_i}}\otimes \rho_i $ for any input state $\rho_i$. Here we treat the environment system as $E \cong B_i\otimes A_i$. According to the discussion above, the coherent information of the GDS channel $\N$ is achieved by the block-diagonal state $\rho = \oplus_{i=0}^n\, p_i\rho_i $, where $\rho_i$ is the input state for the subchannel $\N_i$ and $\bds{p}$ is a probability vector. The coherent information $\Q(\N)$ can be bounded from below as
\begin{equation}
  \begin{aligned}
    \Q(\N) & \geq \log \Big[\sum_i 2^{\Q(\N_i)}\Big] - \frac{\log (n+1)}{2} \min_{\rho_i: \text{opt}} \max_{i,j} ||\omega_i-\omega_j||_1 \\
    & \geq \log(n+1)-\frac{\log(n+1)}{2}\min_{|\phi_i\rangle,|\phi_j\rangle} \max_{i,j} \Big \| \frac{\mathds{1}_{B_i}}{d_{B_i}}\otimes |\phi_i\rangle\langle\phi_i| - \frac{\mathds{1}_{B_j}}{d_{B_j}}\otimes |\phi_j\rangle\langle\phi_j| \Big\|_1 \\
    & \geq \log(n+1)-\frac{\log(n+1)}{2} \max_{i,j} \Big \| \frac{\mathds{1}_{B_i}}{d_{B_i}}\otimes |0\rangle\langle 0|_{A_i} - \frac{\mathds{1}_{B_j}}{d_{B_j}}\otimes |0\rangle\langle 0|_{A_j} \Big\|_1 \\
    & \geq \log(n+1)-\frac{\log(n+1)}{2} \max_{i,j} \Big(\Big|\frac{1}{d_{B_i}} - \frac{1}{d_{B_j}}\Big| + \frac{d_{B_i}-1}{d_{B_i}} + \frac{d_{B_j}-1}{d_{B_j}}\Big) \\
    & = \log(n+1) \min_{i,j} \frac{1}{2}\Big(\frac{1}{d_{B_i}}+ \frac{1}{d_{B_j}} - \Big|\frac{1}{d_{B_i}} - \frac{1}{d_{B_j}}\Big| \Big) \\
    & = \min_i \frac{\log(n+1)}{d_{B_i}} > 0 \,, 
  \end{aligned}
\end{equation}
where the forth inequality is for $A_i \neq A_j$, and we consider the worst case that $ \frac{\mathds{1}_{B_i}}{d_{B_i}}\otimes |0\rangle\langle 0|_{A_i}$ and $\frac{\mathds{1}_{B_j}}{d_{B_j}}\otimes |0\rangle\langle 0|_{A_j}$ have only one non-zero entry located at the same position. 

Therefore, the coherent information $\Q(\N)$ is always positive and is bounded from below by the number of completely depolarizing subchannels, $n+1$, and the maximal dimensions $\max_i d_{B_i}$ of the output subsystem of these subchannels. It is natural to ask whether the quantum capacity of the GDS channel $\N$ is also controlled by the number of subchannels $n+1$ and the output dimensions $d_{B_i}$ of these subchannels. We now provide a positive answer to this question as follows.

\begin{thm}\label{UPQCDCSM}
    Given a fixed integer $d$, suppose that $\N$ is a GDS channel generated by completely depolarizing subchannels $\{\N_i\}_{i=0}^n$ as in Eq.~\eqref{CDCsubchannels}, with the product of input and output dimension is fixed $d_{A_i}d_{B_i} = d$. The quantum capacity of the GDS channel is upper bounded by 
    \begin{equation}
        Q(\N) \leq \log \bigg[1+ 2^{1/4}\, n\max_{i,j}\min\bigg\{\sqrt{\frac{d_{B_i}}{d_{B_j}}},\, \sqrt{\frac{d_{B_j}}{d_{B_i}}}\bigg\} \bigg]\,.
    \end{equation}
    In particular, if for all $d_{B_j}$, and all $d_{B_i}$ satisfying $d_{B_i}>d_{B_j}$ are exactly divisible by $d_{B_j}$, then the coefficient $2^{1/4}$ in the upper bound above can be optimized to be $1$.
\end{thm}
We will use Theorem~\ref{UpBdQNSM} to show this result. First, we introduce some notations. Let $\lfloor x\rfloor$ be the floor function, i.e., $\lfloor x\rfloor$ is the largest integer no greater than $x$; and let $a \pmod b$ denote the remainder of $a$ divided by $b$. For the subchannel $\N_i$ with output dimension $d_{B_i}$, we can rewrite its Kraus operators as
\begin{equation}\label{KrausCDC}
    E_{k}^{(i)} = \frac{1}{\sqrt{d_{B_i}}}\, \big|k \bmod d_{B_i}\big\rangle_{B_i A_i}\big\langle \left\lfloor k/d_{B_i}\right\rfloor\big|\,, \quad k = 0,1,\cdots, d-1\,.
\end{equation}
Here, we do not use the more common expressions $\frac{1}{\sqrt{d_{B_i}}} |s\rangle\langle t|$, since the Kraus operator of $\N$ is given by $E_k = E_k^{(0)}\oplus \cdots \oplus E_{k}^{(n)}$, and we want to emphasize the same index $k$ for all Kraus operators of subchannels.

Now, recall the Choi matrix of $\M_{ij}$ is given in Eq.~\eqref{ChoiMSM}, its partial transposition is
\begin{equation*}
\mathscr{C}_{\M_{ij}}^{T_B} = \sum_{s=0}^{d_{A_i}-1}\sum_{t=0}^{d_{A_j}-1} |s\rangle\langle t|\otimes \sum_k \overline{E_{k}^{(j)}}|t\rangle\langle s|E_{k}^{(i),T}\,,
\end{equation*}
where $\overline{E}$ is the entrywise complex conjugate of the operator $E$, that is $(\overline{E})_{ij} = \overline{E_{ij}}$. Then, we have
\begin{align}
    |\mathscr{C}_{\M_{ij}}^{T_B}|^2 & = \mathscr{C}_{\M_{ij}}^{T_B,\dagger}\, \mathscr{C}_{\M_{ij}}^{T_B} \nonumber \\
    & = \sum_{s,s'}\sum_{t,t'}\sum_{k,k'} |t'\rangle\langle s'|s\rangle\langle t| \otimes \Big(\overline{E_{k'}^{(i)}}|s'\rangle\langle t'|E_{k'}^{(j),T}\overline{E_{k}^{(j)}}|t\rangle\langle s|E_{k}^{(i),T}\Big) \nonumber \\
    & = \sum_{k,k'}\sum_{t,t'} \Big(\langle t'|E_{k'}^{(j),T}\overline{E_{k}^{(j)}}|t\rangle\Big) |t'\rangle\langle t| \otimes \sum_s \overline{E_{k'}^{(i)}}|s\rangle\langle s|E_{k}^{(i),T} \nonumber \\
    & = \sum_{k,k'} \Big(E_{k'}^{(j),T}\overline{E_{k}^{(j)}}\Big)\otimes\Big( \overline{E_{k'}^{(i)}}\,E_{k}^{(i),T}\Big) \nonumber \\
    & = \Big(\sum_k E_{k}^{(j),T} \otimes \overline{E_{k}^{(i)}} \Big)\Big(\sum_k E_{k}^{(j),T} \otimes \overline{E_{k}^{(i)}} \Big)^\dagger \,.
\end{align}
Using this equality and the Kraus operators of the subchannel $\N_i$ given in Eq.~\eqref{KrausCDC}, we obtain that 
\begin{equation}
    |\mathscr{C}_{\M_{ij}}^{T_B}|^2 = \sum_{\substack{k,k'\\ k \bmod d_{B_j} = k'\bmod d_{B_j}\\ \lfloor k/d_{B_i}\rfloor = \lfloor k'/d_{B_i}\rfloor}} \frac{1}{d_{B_i}d_{B_j}} \big|\lfloor k/d_{B_j}\rfloor \big\rangle\big\langle \lfloor k'/d_{B_j} \rfloor\big|\otimes \big|k\bmod d_{B_i}\big\rangle\big\langle k'\bmod d_{B_i}\big|\,.
\end{equation}

\begin{proof}[Proof of Theorem~\ref{UPQCDCSM}]
     We need to carefully analyse the exact expression of $|\mathscr{C}_{\M_{ij}}^{T_B}|^2$ in order to compute $||\tr_B(|\mathscr{C}_{\M_{ij}}^{T_B}|)||_\infty$. According to the conditions $k \bmod d_{B_j} = k'\bmod d_{B_j}$ and $\lfloor k/d_{B_i}\rfloor = \lfloor k'/d_{B_i}\rfloor$, we have
    \begin{equation*}
        k = m_1d_{B_j} + r = q\, d_{B_i} + n_1\,,\quad k' = m_2 d_{B_j} + r = q\, d_{B_i} + n_2\,, \quad n_1, n_2\in \{0,1,\cdots,d_{B_i}-1\}.
    \end{equation*}
    Therefore, 
    \begin{equation}\label{ConditionEq}
        (m_1-m_2)\, d_{B_j} = n_1-n_2 \in \{1-d_{B_i},\cdots,0,\cdots,d_{B_i}-1\}.
    \end{equation}
    Two cases arise, depending on the relationship between $d_{B_j}$ and $d_{B_i}$.
    \begin{itemize}
        \item \textbf{case 1: $d_{B_i}-1<d_{B_j}$}. \\
        Since all these subchannels are distinct, which implies $d_{B_i} \neq d_{B_j}$, hence $d_{B_i}-1<d_{B_j}$ is equivalent to $d_{B_i}<d_{B_j}$. In this situation, for a fixed $m_1$, the only solution of $m_2$ in Eq.~\eqref{ConditionEq} is $m_2 = m_1$, that is $k' = k$ for every fixed $k$. Hence we have 
        \begin{equation}\label{TM1}
        \begin{aligned}
            |\mathscr{C}_{\M_{ij}}^{T_B}|^2 & = \sum_k \frac{1}{d_{B_i}d_{B_j}} \big|\lfloor k/d_{B_j}\rfloor \big\rangle\big\langle \lfloor k/d_{B_j} \rfloor\big|\otimes \big|k\bmod d_{B_i}\big\rangle\big\langle k\bmod d_{B_i}\big| \\
            & \leq \frac{\mathds{1}_{A_j}}{d_{B_i}d_{B_j}} \otimes \mathds{1}_{B_i} \, \big\lceil \frac{d_{B_j}}{d_{B_i}}\big\rceil  \,,
        \end{aligned}
        \end{equation}
        where $\lceil x \rceil$ is the ceil function, i.e., $\lceil x \rceil$ is the smallest integer bigger than $x$. Furthermore, in this case, we have 
        \begin{equation}\label{Bdinf1}
            ||\tr_B(|\mathscr{C}_{\M_{ij}}^{T_B}|)||_\infty \leq \sqrt{\frac{d_{B_i}}{d_{B_j}} \big\lceil \frac{d_{B_j}}{d_{B_i}}\big\rceil} \quad .
        \end{equation}
        \item \textbf{case 2: $d_{B_i}-1 \geq d_{B_j}$}. \\
        For a fixed $r \in \{0,1,\cdots,d_{B_j}-1\}$, define a class of integers $a_{q}^{(n)}$ for $q = 0,\cdots, d_{A_i}-1$ satisfying 
        $$\frac{(q+1)d_{B_i}-r-d_{B_j}}{d_{B_j}} \leq a_q^{(r)} \leq \frac{(q+1)d_{B_i}-1-r}{d_{B_j}}\,.$$
        In fact, $ a_q^{(r)} = \Big\lfloor \frac{(q+1)d_{B_i}-1-r}{d_{B_j}}\Big\rfloor$. Define a class of sets $\mathbf{S}_0^{(r)} = \{0,\cdots,a_0^{(r)}\}$ and $\mathbf{S}_{q+1}^{(r)} = \{a_{q}^{(r)}+1,\cdots,a_{q+1}^{(r)}\}$. These sets have a nice property, \textit{if $k = m_1d_{B_j} + r$ with $m_1 \in \mathbf{S}_{q}^{(r)}$, the index $k' = m_2d_{B_j} + r $ satisfying $k \bmod d_{B_j} = k'\bmod d_{B_j}$ and $\lfloor k/d_{B_i}\rfloor = \lfloor k'/d_{B_i}\rfloor$, implies that $ m_2 \in \mathbf{S}_{q}^{(r)}$.}

        Therefore, we have 
        \begin{equation}\label{TM2}
            \begin{aligned}
                |\mathscr{C}_{\M_{ij}}^{T_B}|^2 & = \sum_{r=0}^{d_{B_j}-1} \frac{1}{d_{B_i}d_{B_j}} \sum_{q = 0}^{d_{A_i}-1} \big(\big|\mathbf{S}_q^{(r)}\big|\big)\, |\Phi_q^{(r)}\rangle\langle \Phi_q^{(r)}| \,,
            \end{aligned}
        \end{equation}
        where $\big|\mathbf{S}_q^{(r)}\big|$ is the number of entries in $\mathbf{S}_q^{(r)}$, $|\Phi_q^{(r)}\rangle = \Big(\sum_{m\in \mathbf{S}_q^{(r)}}  |m,\, (md_{B_j}+r) \bmod d_{B_i}\rangle\Big)\Big/\sqrt{\big|\mathbf{S}_q^{(r)}\big|}$ is the maximal entangled state on the subsystems $A_jB_i$, and these entangled states are mutually orthogonal. Furthermore, in this case, according to Lemma~\ref{Sqr} below about the estimation of $\big|\mathbf{S}_q^{(r)}\big|$, we have 
        \begin{equation}\label{Bdinf2}
        \begin{aligned}
        \tr_B(|\mathscr{C}_{\M_{ij}}^{T_B}|) & = \sum_{r=0}^{d_{B_j}-1} \sum_{q = 0}^{d_{A_i}-1} \frac{\sum_{m\in \mathbf{S}_q^{(r)}}\, |m\rangle\langle m|}{\sqrt{d_{B_i}d_{B_j} \big|\mathbf{S}_q^{(r)}\big| } } \,, \\
            ||\tr_B(|\mathscr{C}_{\M_{ij}}^{T_B}|)||_\infty & \leq \sqrt{\frac{d_{B_j}}{d_{B_i}\Big\lfloor  \frac{d_{B_i}}{d_{B_j}} \Big\rfloor }}\,.
        \end{aligned}
        \end{equation}
    \end{itemize}
    According to Theorem~\ref{UpBdQNSM}, we need to compute the product $||\tr_B(|\mathscr{C}_{\M_{ij}}^{T_B}|)||_\infty||\tr_B(|\mathscr{C}_{\M_{ji}}^{T_B}|)||_\infty$. Since $||\tr_B(|\mathscr{C}_{\M_{ji}}^{T_B}|)||_\infty$ is just exchange the index $i,j$ in $||\tr_B(|\mathscr{C}_{\M_{ij}}^{T_B}|)||_\infty$, which means if $||\tr_B(|\mathscr{C}_{\M_{ij}}^{T_B}|)||_\infty$ follows Eq.~\eqref{Bdinf1}, then $||\tr_B(|\mathscr{C}_{\M_{ji}}^{T_B}|)||_\infty$ follows Eq.~\eqref{Bdinf2}, vice versa. Totally, without loss of generality, assume $d_{B_i}<d_{B_j}$, we have 
    \begin{equation*}
        ||\tr_B(|\mathscr{C}_{\M_{ij}}^{T_B}|)||_\infty||\tr_B(|\mathscr{C}_{\M_{ji}}^{T_B}|)||_\infty \leq \sqrt{\frac{d_{B_i}}{d_{B_j}} \big\lceil \frac{d_{B_j}}{d_{B_i}}\big\rceil} \sqrt{\frac{d_{B_i}}{d_{B_j}\Big\lfloor  \frac{d_{B_j}}{d_{B_i}} \Big\rfloor }} \leq \sqrt{2}\frac{d_{B_i}}{d_{B_j}} \,.
    \end{equation*}
    And further
    \begin{equation*}
        Q(\N) \leq \log \Big(1+ 2^{1/4}\, n\max_{i,j}\min\bigg\{\sqrt{\frac{d_{B_i}}{d_{B_j}}},\, \sqrt{\frac{d_{B_j}}{d_{B_i}}}\bigg\} \Big)\,.
    \end{equation*}
    In particular, for all $d_{B_i}$, and for all $d_{B_j}$ satisfying $d_{B_j}>d_{B_i}$, $d_{B_j}$ is exactly divisible by $d_{B_i}$, then the coefficient  $2^{1/4}$ in the upper bound above can be optimized to $1$.
\end{proof}

\begin{lem}[Estimation of $\big|\mathbf{S}_q^{(r)}\big|$]\label{Sqr}
    If $d_{B_i} - 1 \geq d_{B_j}$, then for all $q$ and $r$, we have 
    \begin{itemize}
        \item \textbf{If $d_{B_i}$ is exactly divisible by $d_{B_j}$}, then $\big|\mathbf{S}_{q}^{(r)}\big| = \frac{d_{B_i}}{d_{B_j}}$. In particular,  for $q = 0,\cdots,d_{A_i}-1$,
$$\mathbf{S}_q^{(r)} = \Big\{q\frac{d_{B_i}}{d_{B_j}},\cdots,(q+1)\frac{d_{B_i}}{d_{B_j}}-1\Big\}  \,, $$
which is independent of $r$.
        \item \textbf{If $d_{B_i}$ is not exactly divisible by $d_{B_j}$}, then $ \Big\lfloor  \frac{d_{B_i}}{d_{B_j}} \Big\rfloor  \leq \big|\mathbf{S}_{q}^{(r)}\big| \leq \Big\lfloor  \frac{d_{B_i}}{d_{B_j}} \Big\rfloor + 1\,.
        $
    \end{itemize}
\end{lem}
\begin{proof}
    Since $\big|\mathbf{S}_0^{(r)}\big| = a_0^{(r)} + 1 = \Big\lfloor  \frac{d_{B_i}-1-r}{d_{B_j}} \Big\rfloor+1$ and $\big|\mathbf{S}_{q+1}^{(r)}\big| = a_{q+1}^{(r)} - a_q^{(r)}$, let $(x)$ be the decimal part of $x$, we have
    \begin{align*}
         a_{q+1}^{(r)} - a_q^{(r)} & = \Big\lfloor  \frac{(q+2)d_{B_i}-1-r}{d_{B_j}} \Big\rfloor - \Big\lfloor  \frac{(q+1)d_{B_i}-1-r}{d_{B_j}} \Big\rfloor \\
         & = \Big\lfloor  \frac{d_{B_i}}{d_{B_j}} \Big\rfloor + \Big\lfloor \Big(\frac{(q+1)d_{B_i}-1-r}{d_{B_j}}\Big) + \Big(\frac{d_{B_i}}{d_{B_j}}\Big)   \Big\rfloor \,.
    \end{align*}
    \begin{itemize}
        \item \textbf{If $d_{B_i}$ is exactly divisible by $d_{B_j}$}, then $\big|\mathbf{S}_{q+1}^{(r)}\big| =  a_{q+1}^{(r)} - a_q^{(r)} = \frac{d_{B_i}}{d_{B_j}} = a_0^{(r)} + 1 = \big|\mathbf{S}_0^{(r)}\big|$ holds for all $q$ and $r$.
        \item  \textbf{If $d_{B_i}$ is not exactly divisible by $d_{B_j}$}, then 
        $\Big\lfloor \frac{d_{B_i}}{d_{B_j}} \Big\rfloor  \leq \big|\mathbf{S}_{q+1}^{(r)}\big| \leq \Big\lfloor  \frac{d_{B_i}}{d_{B_j}} \Big\rfloor + 1 = \big|\mathbf{S}_{0}^{(0)}\big|\,$.
        On the other hand, for $r\geq 1$, we also have $\Big\lfloor  \frac{d_{B_i}}{d_{B_j}} \Big\rfloor\leq  \big|\mathbf{S}_{0}^{(r)}\big|\leq \Big\lfloor  \frac{d_{B_i}}{d_{B_j}} \Big\rfloor +1$.
    \end{itemize}
\end{proof}

\subsection{B. Bounds on Private and Classical Capacities}
We now provide some upper bounds on the private and classical capacities of the GDS channel $\N$, which is generated by completely depolarizing subchannels $\{\N_i\}_{i=0}^n$ in Eq.~\eqref{CDCsubchannels} with Kraus operators in Eq.~\eqref{KrausCDC}. 

First, since the optimal ensemble $\mathfrak{E}'_i$ for $\N_i$ can be chosen as $\mathfrak{E}'_i = \{\mathds{1}_{A_i}/d_{A_i}\}$ such that $I_p(\mathfrak{E}'_i,\N_i) = P^{(1)}(\N_i) = 0$, and hence $\omega_i = \N_i^c(\rho_i) = \mathds{1}/d$ for all $i = 0,\cdots, n$, it holds that 
\begin{equation}
    P^{(1)}(\N) \geq \log(n+1)\,,
\end{equation}
according to Proposition~\ref{PropLBdP1}. The upper bound on the classical capacity we used is derived in Refs.~\cite{SDPBoundC,SDPBoundC2}.
\begin{lem}[\cite{SDPBoundC,SDPBoundC2}]
  For any quantum channel $\B$, its classical capacity is upper bounded as
  \begin{equation}
    C(\B) \leq \log \beta(\B)\,,
  \end{equation}
  where $\beta(\B)$ is the solution of following SDP:
  \begin{equation}
    \beta(\B) = \min\{\tr\, S_{B} | Y_{AB} \pm \mathscr{C}_{\B}^{T_B} \geq 0\,, \mathds{1}_A \otimes S_{B} \pm Y_{AB}^{T_B} \geq 0,\ Y_{AB}^\dagger = Y_{AB},S_{B}^\dagger = S_{B}\}\,.
  \end{equation}
\end{lem}
We now provide an upper bound on the classical capacity of the GDS channel $\N$. In particular, we require that for all $d_{B_j}$, and all $d_{B_i}$ satisfying $d_{B_i}>d_{B_j}$ are exactly divisible by $d_{B_j}$, which gives that $\big|\mathbf{S}_{q}^{(r)}\big| = \frac{d_{B_i}}{d_{B_j}}$ for all $q$ and $r$.
\begin{thm}\label{PCNSM}
    Given a fixed integer $d$, suppose that $\N$ is a GDS channel generated by completely depolarizing subchannels $\{\N_i\}_{i=0}^n$ as Eq.~\eqref{CDCsubchannels}, with the product of input and output dimensions is fixed: $d_{A_i}d_{B_i} = d$. In particular, suppose that for all $d_{B_j}$, and all $d_{B_i}$ satisfying $d_{B_i}>d_{B_j}$ are exactly divisible by $d_{B_j}$. The private and classical capacities of the GDS channel are weakly additive, and equal to $\log (n+1)$,
    \begin{equation}
       P^{(1)}(\N) = P(\N) = C^{(1)}(\N) = C(\N) = \log (n+1).
    \end{equation}
\end{thm}
\begin{proof}
    Here we choose $Y_{AB}$ as 
    \begin{equation}\label{YABC}
        Y_{AB} = \sum_{i=0}^n Y_{A_iB_i} + \sum_{i<j} \big(|\mathscr{C}_{\M_{ij}}^{T_B}| + |\mathscr{C}_{\M_{ij}^\dagger}^{T_B}|\big)\,,
    \end{equation}
    which satisfies $Y_{AB} \pm \mathscr{C}_{\B}^{T_B} \geq 0$. In particular, $Y_{A_iB_i} = \mathds{1}_{A_i} \otimes \frac{\mathds{1}_{B_i}}{d_{B_i}} = Y_{A_iB_i}^{T_B}$.  The expression of $|\mathscr{C}_{\M_{ij}}^{T_B}|$ is shown in Eq.~\eqref{TM1} and~\eqref{TM2}.
    \begin{itemize}
        \item \textbf{Case 1: $d_{B_i}-1<d_{B_j}$}. We have 
        \begin{equation}
            |\mathscr{C}_{\M_{ij}}^{T_B}|^{T_B} \leq \sqrt{\frac{\frac{d_{B_j}}{d_{B_i}}}{d_{B_i}d_{B_j}}} \quad  \mathds{1}_{A_j} \otimes \mathds{1}_{B_i} = \mathds{1}_{A_j} \otimes \frac{\mathds{1}_{B_i}}{d_{B_i}}\,. 
        \end{equation}
         \item  \textbf{Case 2: $d_{B_i}-1 \geq d_{B_j}$}. We have
         \begin{equation*}
            \begin{aligned}
                |\mathscr{C}_{\M_{ij}}^{T_B}| & = \sum_{r=0}^{d_{B_j}-1} \sum_{q = 0}^{d_{A_i}-1} \frac{1}{d_{B_j}}\, |\Phi_q^{(r)}\rangle\langle \Phi_q^{(r)}| \,,
            \end{aligned}
        \end{equation*}
        where $|\Phi_q^{(r)}\rangle = \Big(\sum_{m\in \mathbf{S}_q^{(r)}}  |m,\, (md_{B_j}+r) \bmod d_{B_i}\rangle\Big)\Big/\sqrt{\frac{d_{B_i}}{d_{B_j}}}$ is the maximal entangled state on the subsystems $A_jB_i$. According to the basic inequality $\text{SWAP} \leq \mathds{1}$, where $\text{SWAP} = \sum_{s,t} |st\rangle\langle ts|$ is the SWAP operator, it holds
        \[
        |\Phi_q^{(r)}\rangle\langle \Phi_q^{(r)}|^{T_B} \leq \frac{d_{B_j}}{d_{B_i}} \Big(\sum_{m\in \mathbf{S}_q^{(r)}} |m\rangle\langle m|\Big) \otimes \Big(\sum_{m'\in \mathbf{S}_q^{(r)}} |(m'd_{B_j}+r) \bmod d_{B_i}\rangle\langle (m'd_{B_j}+r) \bmod d_{B_i}|\Big)\,.
        \]
        Notice that by the definition of $\mathbf{S}_q^{(r)}$, when $d_{B_i}$ is exactly divisible by $d_{B_j}$, then for $q = 0,\cdots,d_{A_i}-1$,
        $$\mathbf{S}_q^{(r)} = \{q\frac{d_{B_i}}{d_{B_j}},\cdots,(q+1)\frac{d_{B_i}}{d_{B_j}}-1\}\,,$$
        which is independent of $r$. As $\mathbf{S}_q^{(r)} \cap \mathbf{S}_{q'}^{(r)} = \emptyset$ if $q\neq q'$, and the union of sets $\cup_{q=0}^{d_{A_i-1}}\, \mathbf{S}_q^{(r)} = \{0,\cdots,d_{A_j}-1\}$, we obtain that 
        \begin{equation}
        \begin{aligned}
            |\mathscr{C}_{\M_{ij}}^{T_B}|^{T_B} & \leq \frac{1}{d_{B_i}} \sum_{r=0}^{d_{B_j}-1} \sum_{q = 0}^{d_{A_i}-1} \Big(\sum_{m\in \mathbf{S}_q^{(r)}} |m\rangle\langle m|\Big) \otimes \Big(\sum_{m'\in \mathbf{S}_q^{(r)}} |(m'd_{B_j}+r) \bmod d_{B_i}\rangle\langle (m'd_{B_j}+r) \bmod d_{B_i}|\Big) \\
            & = \frac{1}{d_{B_i}} \sum_{q = 0}^{d_{A_i}-1} \Big(\sum_{m\in \mathbf{S}_q^{(r)}} |m\rangle\langle m|\Big) \otimes \Big( \sum_{r=0}^{d_{B_j}-1} \sum_{m'\in \mathbf{S}_q^{(r)}} |(m'd_{B_j}+r) \bmod d_{B_i}\rangle\langle (m'd_{B_j}+r) \bmod d_{B_i}|\Big) \\
            & = \frac{1}{d_{B_i}} \sum_{q = 0}^{d_{A_i}-1} \Big(\sum_{m\in \mathbf{S}_q^{(r)}} |m\rangle\langle m|\Big) \otimes \mathds{1}_{B_i} \\
            & = \mathds{1}_{A_j} \otimes \frac{\mathds{1}_{B_i}}{d_{B_i}} \,.
        \end{aligned}
        \end{equation}
    \end{itemize}
    Therefore, we have $|\mathscr{C}_{\M_{ij}}^{T_B}|^{T_B} + |\mathscr{C}_{\M_{ij}^\dagger}^{T_B}|^{T_B} \leq \mathds{1}_{A_i} \otimes \frac{\mathds{1}_{B_j}}{d_{B_j}}+ \mathds{1}_{A_j} \otimes \frac{\mathds{1}_{B_i}}{d_{B_i}}$ for all $i\neq j$. By choosing $S_B = \bigoplus_{i=0}^n \frac{\mathds{1}_{B_i}}{d_{B_i}}$, we find that
    \[
    \mathds{1}_A\otimes S_B = \sum_{i=0}^n \mathds{1}_{A_i} \otimes \frac{\mathds{1}_{B_i}}{d_{B_i}} + \sum_{i<j} \Big( \mathds{1}_{A_i} \otimes \frac{\mathds{1}_{B_j}}{d_{B_j}}+ \mathds{1}_{A_j} \otimes \frac{\mathds{1}_{B_i}}{d_{B_i}}\Big) \geq \pm Y_{AB}^{T_B} \,.
    \]
    Furthermore, $\tr\, S_B = n+1$ implies that $C(\N) \leq \log (n+1)$. Together with $P^{(1)}(\N) \geq \log (n+1)$, we proved the theorem as desired.
\end{proof}
One can also relax the divisibility requirement of dimensions to obtain an upper bound on the classical capacity for generic GDS channels, which is close to $\mathcal{O}(\log (n+1))$. Here we omit such a result since it doesn't give more insight than Theorem~\ref{PCNSM}.

Now we give an example to show that the condition in Theorem~\ref{PCNSM} can be achieved exactly. 
\begin{ex}
    Given two integers $p>1$ and $n$. Let the GDS channel be generated by following completely depolarizing subchannels with Kraus operators in the form of Eq.~\eqref{KrausCDC}
\begin{equation}
\N_i: A_i \to B_i,\,\quad  \N_i(\rho_i) = \frac{\mathds{1}_{B_i}}{d_{B_i}}\,,\quad \text{for}\quad i = 0,\cdots,n\,,
\end{equation}
where $d_{A_i} = p^i$, $d_{B_i} = p^{\alpha - i}$, and $\alpha\geq n$ is an integer. Then for all $d_{B_j}$, all $d_{B_i}$ satisfying $d_{B_i} = p^{\alpha - i}>d_{B_j} = p^{\alpha - j}$ are exactly divisible by $d_{B_j}$, as $d_{B_i}/d_{B_j} = p^{j-i}$ is an integer. Hence, this case satisfies the condition in Theorem~\ref{PCNSM}. According to Theorems~\ref{UPQCDCSM} and~\ref{PCNSM}, we have
\begin{align}
    \frac{\log(n+1)}{p^{\alpha}} \leq \Q(\N) \leq Q(\N) \leq \log \Big(1+\frac{n}{\sqrt{p}}\Big) < \log(n+1) = P(\N) = C(\N)\,.
\end{align}
Since $p$ is arbitrary, two key results are worth noting.
\begin{itemize}
    \item[(1).] For any positive number $\beta>0$, we can take $n/\sqrt{p}$ small enough such that $0<Q(\N) < \beta$. Hence, we generalized the similar result about ``platypus channel" $\M_d$ provided in~\cite{platypusTIT,wu2025}. In fact, as we mentioned, the generalized platypus defined in~\cite{platypusTIT} is also a special GDS channel generated by two subchannels $\N_0(\rho_0) = \sigma$ and $\N_1(\rho_1) = |0\rangle\langle 0|$.
    \item[(2).]  For a fixed $n$, we can regard $p = p(n)$ as a function of $n$. By choosing suitable $p(n)$, for example $p(n) = n^4$, the gap between $P(\N)$ and $Q(\N)$ can tend to infinity. The only known example with such extensive separation is ``half-rocket channels'', whose quantum capacity is between 0.6 and 1, but private capacity equals $\log d$ with the input and output dimensions are $d_A = d^2, d_B = d_E = d^6 - d^4$~\cite{Leung2014}. Here, the quantum capacity of the GDS channels can be arbitrarily close to 0 while the private capacity tends to infinity. More importantly, the construction of half-rocket channels relies on random unitaries, which is more technical and is hard to write the exact expression, while the GDS channels provided here give an exact example with a simpler expression. Very recently, Ref.~\cite{WuPeixue2025} provide a new quantum channel $\M_n = \mathcal{C}^{\otimes n}$ defined by a special channel $\mathcal{C}$. The channel $\M_n$ has a large gap between quantum and private capacity based on private channels~\cite{Horodeckichannel2} and the spin alignment conjecture~\cite{platypusTIT}. Here, the unbound gap between quantum and private capacity for our GDS channels is unconditional.
\end{itemize}
\end{ex}

\end{document}